\documentclass{article} 
\usepackage{iclr2025_conference,times}

\usepackage{amsmath,amsfonts,bm, amsthm, amssymb}

\def\eqref#1{equation~\ref{#1}}

\def\1{\bm{1}}

\DeclareMathAlphabet{\mathsfit}{\encodingdefault}{\sfdefault}{m}{sl}
\SetMathAlphabet{\mathsfit}{bold}{\encodingdefault}{\sfdefault}{bx}{n}

\newcommand{\E}{\mathbb{E}}

\DeclareMathOperator*{\argmin}{arg\,min}

\newcommand{\mL}{\mathcal{L}}

\newcommand{\interior}[1]{%
  {\kern0pt#1}^{\mathrm{o}}%
}

\newcommand{\intsim}{\interior{\Sigma}}
\newcommand{\Xspace}{\mathcal{X}}
\newcommand{\prob}{\mathcal{P}(\Xspace)}
\newcommand{\citetinside}[2]{#1 \citep{#2}}
\newcommand{\measurable}{\mathcal{M}(\Xspace)}
\newcommand{\meanZero}{\mathcal{W}(\Xspace)}
\newcommand{\lebesgueX}{\lambda(\Xspace)}

\newtheorem{theorem}{Theorem}
\newtheorem{example}{Example}
\newtheorem{definition}{Definition}
\newtheorem{lemma}{Lemma}
\newtheorem{corollary}{Corollary}
\newtheorem{assumptions}{Assumption}
\newtheorem{remark}{Remark}
\newtheorem{conjecture}{Conjecture}
\newtheorem{proposition}{Proposition}
\newtheorem{customassumption}{Assumption}

\renewcommand{\emph}[1]{\textit{#1}} 

\usepackage{hyperref}
\usepackage{url}

\usepackage{lipsum}
\usepackage{graphicx}
\usepackage{xcolor}
\usepackage{geometry}
\usepackage{fancyhdr}
\usepackage{multicol}
\usepackage{array}
\usepackage{booktabs}       
\usepackage{nicefrac}       
\usepackage{microtype}      
\usepackage{soul}
\usepackage{cleveref} 
\usepackage{ulem}

\title{Common Functional Decompositions Can Mis-attribute Differences in Outcomes Between Populations}

\author{%
  Manuel Quintero \\ MIT IDSS \\ mquint@mit.edu
  \And
  William T. Stephenson \\ MIT Lincoln Laboratory\thanks{DISTRIBUTION STATEMENT A. Approved for public release. Distribution is unlimited.
} \\ william.stephenson@ll.mit.edu
  \And
  Advik Shreekumar \\ MIT Economics \\ adviks@mit.edu
  \And
  Tamara Broderick \\ MIT EECS \\ tbroderick@mit.edu
}

\renewcommand{\emph}[1]{\textit{#1}} 

\iclrfinalcopy 
\begin{document}
\maketitle
\begin{abstract}
    In science and social science, we often wish to explain why an outcome is different in two populations. 
    For instance, if a jobs program benefits members of one city more than another, is that due to differences in program participants (particular covariates) or the local labor markets (outcomes given covariates)? The Kitagawa-Oaxaca-Blinder (KOB) decomposition is a standard tool in econometrics that explains the difference in the mean outcome across two populations. 
    However, the KOB decomposition assumes a linear relationship between covariates and outcomes, while the true relationship may be meaningfully nonlinear. 
    Modern machine learning boasts a variety of nonlinear functional decompositions for the relationship between outcomes and covariates in one population. 
    It seems natural to extend the KOB decomposition using these functional decompositions. 
    We observe that a successful extension should not attribute the differences to covariates — or, respectively, to outcomes given covariates — if those are the same in the two populations. 
    Unfortunately, we demonstrate that, even in simple examples, two common decompositions — functional ANOVA and Accumulated Local Effects — can attribute differences to outcomes given covariates, even when they are identical in two populations. 
    We provide a characterization of when functional ANOVA misattributes, as well as a general property that any discrete decomposition must satisfy to avoid misattribution. 
    We show that if the decomposition is independent of its input distribution, it does not misattribute. 
    We further conjecture that misattribution arises in any reasonable additive decomposition that depends on the distribution of the covariates.
\end{abstract}

\section{Introduction}
\label{Sec:intro}
\textbf{Motivating Examples.} 
\begin{enumerate}
    \item A worker at a government health department is reviewing patient mortality rates (\(Y\)) at two hospitals, H and K.
    He notices that the mortality rate is lower in hospital K (\(\mathbb{E}_{K}[Y] < \mathbb{E}_{H}[Y]\)).    
    Is mortality lower because K receives patients who are easier to treat?
    Or is K more effective at providing care?
    If he can determine which explanation is more accurate, he may be able to better allocate training or resources across the two hospitals.
    
    \item The mayor of City K compares the results of a job training program to a similar one in City H.     
    She notices that program graduates in her city have lower post-program income (\(Y\)) than those in City H (\(\mathbb{E}_{K}[Y] < \mathbb{E}_{H}[Y]\)). 
    Is income lower because program graduates in City K are meaningfully different from those in City H?
    Or do jobs in City K tend to pay workers less than jobs in City H?
    If she can figure out why this difference occurs, perhaps she can modify the job training program or its recruitment strategy to make it more effective.    
\end{enumerate}
Many scientific questions reduce to similar comparisons between two populations. 
After observing differences, analysts often want to ask why these differences occur. 
One reason might be that the populations differ on meaningful traits. 
In our first example, perhaps the distribution of \(X\) is unequal: say, both hospitals provide the same standard of care, but hospital K's patient population is healthier and naturally has lower mortality rates. 
In this case, covariates \(X\) drive the difference. 
Or perhaps the patients in both hospitals are similar, but the medical staff in hospital K are particularly skilled at treating serious conditions, such as pneumonia or heart attacks. 
In this case, outcomes given covariates \(Y \mid X\) drive the difference. 
A useful explanation for mean differences between populations would distinguish between these possibilities, as well as describe which aspects of the covariates or outcomes given covariates explain the difference.

The Kitagawa-Oaxaca-Blinder (KOB) decomposition \citep{kitagawaComponentsDifferenceTwo1955, oaxacaMaleFemaleWageDifferentials1973, blinderWageDiscriminationReduced1973} is widely used in the econometrics literature to solve exactly this problem. 
The KOB decomposition separates a difference of means into components that depend on the distribution of covariates, \(X\), and those that depend on the conditional expectation of outcomes given covariates, \(\mathbb{E}[Y | X]\). 
However, it relies on parametric linear models of \(\mathbb{E}[Y | X]\), which are likely inadequate to describe the complex and heterogeneous relationships that may arise in practice \citep{fortinDecompositionMethodsEconomics2011,bach2024heterogeneity}. 
A natural generalization of the KOB decomposition would allow for non-linear or nonparametric models for \(\mathbb{E}[Y | X]\). 
Such an approach could account for shifts in the distribution of \(X\) through a generic step-wise transformation that moves the distribution of \(X\) from population H to population K. 
Importance can be assigned to individual features in the change in conditional expectation \(\mathbb{E}[Y | X]\) through the use of additive functional decomposition methods.

Such a generalization requires a choice of the functional decomposition.
Fortunately, modern machine learning offers multiple options. 
For example, the functional ANOVA (FANOVA) \citep{stoneUsePolynomialSplines1994,huangProjectionEstimationMultiple1998,hooker2004discovering,hookerGeneralizedFunctionalANOVA2007,agrawalSKIMFAKernelHighDimensional2021} and Accumulated Local Effects (ALE) \citep{apleyVisualizingEffectsPredictor2020} decompositions have been widely used in sensitivity analysis \citep{chastaingGeneralizedHoeffdingSobolDecomposition2012, antoniadisRandomForestsGlobal2021}, machine learning interpretability \citep{lengerichPurifyingInteractionEffects, limmerNeuralANOVA2024}, finance \citep{liang2022timesequencing, belhadi2021ensemble}, and environmental and climate sciences \citep{huangBeyondPredictionIntegrated2023, peichlMachinelearningMethodsAssess2021, hillModellingClimateUsing2023}. 

The success of such decompositions makes them seem like natural choices for use in explaining a difference in means. 
However, we demonstrate that common functional decompositions---the FANOVA and ALE---are ill-suited for this task in that they can misattribute differences stemming from a changing \(X\) to differences from changing \(Y \mid X\). 
Throughout, we will use ``misattribution'' as a shorthand for differences in \(X\) attributed to \(Y \mid X\).
We characterize when FANOVA makes this misattribution and provide a characterization of when a general decomposition will misattribute in discrete settings (i.e., when the covariate space is countable).
In particular, we show that misattribution occurs whenever the functional decomposition depends on the input covariate distribution.
We conjecture and partially prove that this result holds in continuous settings as well.

The remainder of this paper is structured as follows. 
In Section~\ref{sec:related}, we review the Kitagawa-Oaxaca-Blinder (KOB) decomposition and discuss functional decomposition methods commonly used in machine learning, such as FANOVA and ALE. 
In Section~\ref{sec:kob_extension}, we define our generalized decomposition framework for difference in means, extending the KOB decomposition to non-linear models. 
In Section~\ref{sec:failure}, we show through simple examples that FANOVA and ALE misattribute effects, and we characterize when FANOVA fails in practical cases of interest. 
In Section~\ref{sec:mainthm}, we provide a general characterization of when a functional decomposition misattributes effects. 
Finally, in Section~\ref{sec:conclusion}, we discuss the implications of our findings.

\section{Related work and notation} 
\label{sec:related}

\subsection{Notation}

Throughout this work, we let \(\Xspace \subseteq \mathbb{R}^d\) denote the feature space of the column random vector \(X = (X_1, \dots, X_d)^T\).
In general, we assume \(\Xspace\) is a subset of \(\mathbb{R}^d\) (continuous case), except for when we explicitly state that \(\Xspace\) is countable (discrete case).
We write \(x = (x_1,\dots,x_d)^T \in \Xspace\) for a realization of \(X\). 
Uppercase letters (\(X, X_i\)) thus denote random variables, while lowercase letters (\(x, x_i\)) denote specific realizations of these.
We denote by \(Y \in \mathbb{R}\) the outcome.
When referring to distributions over the joint \((X,Y)\), we use capital letters such as \(H\) or \(K\). 
Subscripts indicate marginals or conditionals: \(H_X\) is the marginal distribution of \(X\) under \(H\); \(H_i\) is the marginal distribution of the \(i\)-th coordinate \(X_i\); and \(H_{1:i}\) is the joint distribution of \((X_1, \dots, X_i)\).
Probability densities or mass functions are denoted by lowercase letters, such as \(h(x)\) or \(k(x)\).
All probability measures are defined on the Borel \(\sigma\)-algebra of \(\Xspace\) in the continuous case or on the power set of \(\Xspace\) in the discrete case.

\subsection{Kitagawa-Oaxaca-Blinder decomposition}
\label{sec:KOB}
The Kitagawa-Oaxaca-Blinder (KOB) decomposition provides a framework for explaining differences in means between two populations by decomposing them into components attributable to differences in covariates and conditional expectations. 
In its original formulation, KOB assumes the covariate space is \(\Xspace = \mathbb{R}^d\) and that the covariates \(X\) have a linear relationship with the outcome \(Y \in \mathbb{R}\):
\[
    \mathbb{E}_{H_{Y \mid X}}[Y \mid X] = X^T \beta_H \quad \text{and} \quad \mathbb{E}_{K_{Y \mid X}}[Y \mid X] = X^T \beta_K,
\]
where \( H_{Y \mid X} \) is the conditional distribution of \( Y \mid X \) in population \(H\), and similarly for population \(K\).
The vectors \( \beta_H,  \beta_K \in \mathbb{R}^d\) are the regression coefficients defining the linear relationship \( \E[Y \mid X] \) for each population.
The KOB decomposes the difference \(\mathbb{E}_K[Y] - \mathbb{E}_H[Y]\) as
\begin{align}
    \label{eqn:generic decomp}
    & \underbrace{\mathbb{E}_{K_X}[\mathbb{E}_{K_{Y \mid X}}[Y \mid X] - \mathbb{E}_{H_{Y \mid X}}[Y \mid X]]}_{\text{\(Y \mid X\) effect}}  
    + \underbrace{\mathbb{E}_{K_X}[\mathbb{E}_{H_{Y \mid X}}[Y \mid X]] - \mathbb{E}_{H_X}[\mathbb{E}_{H_{Y \mid X}}[Y \mid X]]}_{\text{Covariate effect}} \\
    \label{eqn:KOB}
    &= \sum_{j=1}^d \underbrace{\mathbb{E}_{K_X}[X_j] (\beta^K_j - \beta^H_j)}_{\text{\(Y \mid X\) effect for the \(j\)th covariate}} 
      + \sum_{j=1}^d \underbrace{(\mathbb{E}_{K_X}[X_j] - \mathbb{E}_{H_X}[X_j])\beta^H_j}_{\text{Covariate effect for the \(j\)th covariate}}.
    \end{align}

In the next section, we introduce a natural extension of the KOB decomposition that retains a similar interpretation, but allows for more general forms of $Y \mid X$ by using functional decompositions.
The existing literature provides several such functional decompositions; however, we focus only on additive decompositions that decompose functions into additive components, as they provide a natural analogue of the additive form of $Y \mid X$ in the KOB decomposition. 
Other functional decomposition methods, such as Partial Dependence Plots \citep{friedman2001greedy}, do not offer additive decompositions and thus cannot be immediately incorporated into KOB-like decompositions.
We next review two particularly common additive functional decompositions: FANOVA and ALE.

\subsection{FANOVA}
\label{sec:fanova}
FANOVA measures the importance of features in determining the output of a function and in identifying underlying additive interactions between subsets of variables \citep{hooker2004discovering}. 
It provides a natural representation of a functional in terms of low-order components \citep{hookerGeneralizedFunctionalANOVA2007} by stating that a square-integrable function \(f(x)\) with \(x \in \Xspace = \mathbb{R}^d\) can be written uniquely as \(f(x) = \sum_{S \in 2^{[d]}} \mL(f,K_X,S)(x),\) where \(2^{[d]}\) denotes the power set of \( [d] = \{1, 2, \dots, d\} \) and \(K_X\) is a general measure of the covariates. Then, the components are jointly defined as
\begin{equation}
    \label{eq:fanova}
    \{\mL(f,K_X,S)(x) \mid S \in 2^{[d]} \} = \argmin_{ \{\mL(f, K_X, S) \in L^2(\mathbb{R}^d)\}_{S \in 2^{[d]}}} \int \left( \sum_{s \in 2^{[d]}} \mL(f,K_X,S)(x) - f(x) \right)^2 k(x) dx,
\end{equation}
{subject to \textit{hierarchical orthogonality conditions} among the components:
\begin{equation}
\label{eq:HO}
    \forall S \in 2^{[d]}, \, \forall \, V \subsetneq S : \int \mL(f,K_X,S)(x) \mL(f,K_X,V)(x) k(x) dx = 0,
\end{equation}
where \( \subsetneq \) denotes a proper subset.

Note that the FANOVA component corresponding to a subset \(S\) depends only on the covariates in \(S\), as it is constructed to capture their contribution separately from the rest. 
However, for the sake of generality in defining a functional decomposition, we express it as a function of the full covariate vector. 
The same applies to the ALE decomposition below.

\subsection{Accumulated Local Effects (ALE)}
\label{sec:ale}
ALE is another additive functional decomposition method that is particularly suitable for visualizing the effects of predictors \citep{apleyVisualizingEffectsPredictor2020}. 
Although ALE is defined more generally---allowing for non-differentiable 
\( f \) and extending to dimensions \( d > 2 \)---the case for \( d = 2 \) with a differentiable model, \( f(x_1, x_2) = \E[Y \mid X_1 = x_1, X_2 = x_2] \), suffices for our illustrative purposes. 
The ALE main effect component for the first covariate \(X_1\) is then defined as:
\begin{equation}
    \label{eq:ale}
    \mL(f, K_X, \{1\})(x) = \int_{x_{\text{min},1}}^{x_1} \mathbb{E}_{K_2}\left[ \frac{\partial f(z_1, X_2)}{\partial z_1} \right] dz_1 - \text{constant},
    \end{equation}
where \( \frac{\partial f(x_1, x_2)}{\partial x_1} \) denotes the partial derivative of \(f\) with respect to the first component, \( x_{\text{min},1} \) is a lower bound of the support of \(K_1\), and \textit{constant} is a centering constant such that \(\E_{K_X}[\mL(f, K_X, \{1\})(X)] = 0\). The term \(\mL(f, K_X, \{2\})\) is defined similarly; for the definition of \(\mL(f, K_X, \{1,2\})\) and for the \(d > 2\) case, see \citep{apleyVisualizingEffectsPredictor2020}.

\section{Additive decompositions of population differences}
\label{sec:kob_extension}
As discussed in Section~\ref{Sec:intro}, a desirable extension of KOB would allow for arbitrary flexible regression models by extending it to non-linear functional forms. 
Recall the KOB decomposition in Equation~\ref{eqn:KOB} separates a difference in means into a \(Y \mid X\) effect and a covariate effect. 
To extend the KOB decomposition to more flexible models, we assume a general regression model \( f: \Xspace \to \mathbb{R} \) is fitted such that \( f^K(x) = \mathbb{E}_{K_{Y \mid X}}[Y \mid X = x] \), and similarly for population H.\footnote{ We write an equality \( f^K(x) = \mathbb{E}_{K_{Y \mid X}}[Y \mid X = x] \) for the purpose of exposition. In practice, the fitted \( f^K(x) \) will contain approximation error, and our results apply to decompositions of \( f^K(x) \) rather than the exact expectation \( \mathbb{E}_{K_{Y \mid X}}[Y \mid X = x] \).}
Our goal is to decompose Equation~\ref{eqn:generic decomp} into smaller, interpretable components just as in the KOB decomposition.
To achieve this goal and in the spirit of FANOVA and ALE discussed in Section~\ref{sec:related}, we assume a generic additive functional decomposition, denoted by \( \mL \), which operates on arbitrary functions \( f \) of the covariates, probability measures of the covariates \( H_X \), and subsets of features \( S \). 
We assume this decomposition yields an additive representation that holds for all \( x \in \Xspace \subseteq \mathbb{R}^d \):
\begin{equation}
    \label{eq:additive_rep}
    f(x) = \sum_{S \in 2^{[d]}} \mL(f, H_X, S)(x).
\end{equation}

Given such an additive functional decomposition, it is straightforward to extend the KOB decomposition. 
We define two types of swaps, analogous to the terms in the KOB decomposition. 
First, we can swap out the distribution of each one-dimensional covariate of \( X \) at a time; we call such terms \textit{the difference due to change in \(X\)}. 
Second, we can swap out the functional decomposition terms of \( f^H \) for those of \( f^K \); we call such terms \textit{the difference due to change in \( Y \mid X \)}. 
We define this KOB extension below:
\begin{definition} 
\label{def:decomp}
	Let \(\mathcal{S}\) define an ordering of all subsets \(S \in 2^{[d]}\); we refer to the \(i\)th subset in this ordering as \(\mathcal{S}_i\). We define \emph{the importance decomposition} to be:
	\begin{align}
		\mathbb{E}_K[Y] - & \mathbb{E}_H[Y] = \sum_{i=1}^{|\mathcal{S}|} \delta^{Y \mid X}_{S_i} + \sum_{j=1}^{d} \delta^{X}_j, \label{eq:import_decomp}
	\end{align}
	\begin{align*}
		\textrm{where: } & \delta^{Y \mid X}_{S_i} := \mathbb{E}_{H_X} \left[ \mL(f^K, K_X, S_i)(X)\right] - \mathbb{E}_{H_X} \left[ \mL(f^H, H_X, S_i)(X) \right],\\
		&\delta^{X}_j := \mathbb{E}_{K_{1:j | j+1:d} H_{j+1:d}} \left[ f^K(X) \right] - \mathbb{E}_{K_{1:j-1 | j:d} H_{j:d}} \left[ f^K(X) \right].
	\end{align*}
	
Note that \(\delta^{Y \mid X}_{S_i}\) holds the covariate distribution fixed at \(H_X\), and changes whether the \(S_i\) term of \(\mathbb{E}[Y \mid X]\) comes from \(H\) or \(K\).
We therefore call \(\delta^{Y \mid X}_{S_i}\) \emph{the difference due to the dependence of \(Y\mid X\) on feature subset \(S_i\)}. 
Likewise, the distributions over covariates in \(\delta^{X}_j\) differ in whether \(X_j\) follows a distribution determined by \(H\) or \(K\). 
We therefore call \(\delta^{X}_j\) \emph{the difference due to the change in distribution of covariate \(j\)}.
\end{definition} 

Definition~\ref{def:decomp} is an extension of the KOB decomposition from Section~\ref{sec:related}, which also defines differences from swapping out distributions of covariates, as well as differences in swapping out (a model for) \(Y \mid X\). 
The main difference is that Definition~\ref{def:decomp} uses a generic additive decomposition of \(Y \mid X\), whereas the KOB decomposition assumes a linear model.

This decomposition---like the KOB decomposition---makes a series of specific choices: first swapping \(S_1\), then \(S_2\), ... then finally swapping \(S_{\left|\mathcal{S}\right|}\), and then swapping covariate one, then covariate two, etc.
Why not swap \(S_2\) first? Why not swap covariate three immediately after \(S_1\)?
In general, there is no reason to prefer any one ordering, and different orderings will produce different results.
With no preferred ordering of swaps, one may prefer to average over all possible orderings and report the resulting averages as the definitions of \(\delta^{Y \mid X}_{S_i}\) and \(\delta^{X}_{j}\).\footnote{\citet{shorrocksDecompositionProceduresDistributional2013} describes such averages as applying logic of Shapley values to functional decompositions.}
Our results here apply to any fixed order; we leave the extension to averaging over all orderings as future work.

\section{Failure of existing functional decompositions}
\label{sec:failure}

Once a user has specified the functional forms of \(f^H(x)\) and \(f^K(x)\), the only decision to be made before using Definition~\ref{def:decomp} is the choice of functional decomposition \(\mL\). 
At first glance, options such as ALE or FANOVA from Section~\ref{sec:fanova} and Section~\ref{sec:ale} seem like excellent choices: they provide additive decompositions of generic functions with properties that make them well-suited for understanding functions in other applications. 
However, we show that a broad class of functional decompositions, including FANOVA and ALE, are inappropriate for explaining population differences in the sense of Definition~\ref{def:decomp}, despite their great success in other applications. 
In particular, we characterize when such decompositions incorrectly state that differences stem from changes in \(Y \mid X\).

Recall that Definition~\ref{def:decomp} defines \(\delta_{S_i}^{Y \mid X}\) to be the \emph{difference due to the dependence of \(Y \mid X\) on feature subset \(S_i\)}.
Suppose that the distributions of \(Y \mid X\) are in fact identical across the two populations: \(f^K = f^H = f\).
In this situation, any reasonable decomposition should lead us to believe there is no difference due to differences in \(Y \mid X\); that is, \(\delta_{S_i}^{Y \mid X} = 0\). 
Unfortunately, we present examples where FANOVA and ALE can misattribute differences in \(X\) to differences in \(Y \mid X\). 

\subsection{Examples where FANOVA and ALE misattribute}
To begin with, we formalize what we mean by misattribution. Note that when \( f^K = f^H = f \), \( \delta^{Y \mid X}_S \) reduces to 
\begin{equation}
\label{eq:delta}
    \Delta(\mathcal{L}, f,H_X,K_X,S) := \delta^{Y \mid X}_S = \mathbb{E}_{H_X} \left[ \mL(f,K_X,S)(X) - \mL(f,H_X,S)(X) \right].
\end{equation}
With \cref{eq:delta} in mind, we define misattribution as follows:

\begin{definition}
    \label{def:misassign}
    We say that a functional decomposition \(\mL\) \textit{misattributes effects of \( Y \mid X \)} if \(\Delta(\mathcal{L}, f, H_X, K_X, S) \neq 0\) for any \(f, H_X, K_X, S\).
\end{definition}

In the following examples, we show that FANOVA and ALE misattribute effects of \(Y \mid X\) in simple scenarios; we present here a summary of the examples, and the full details are provided in Appendix~\ref{app:examples}.

\begin{example}[FANOVA]
\label{ex:fanova}
Consider the case where the fitted model is additive and has two covariates: \(f^{K} = f^{H} = f(x) = x_1 + x_2\). 
Suppose that population H has covariates \(X_1, X_2\), with \(\mathbb{E}_{H_X}[X_1] = \mathbb{E}_{H_X}[X_2] = 0\) while in population K, \(\mathbb{E}_{K_X}[X_1] = \mu\) and \(\mathbb{E}_{K_X}[X_2] = 0\) for \(\mu \neq 0\). 
In both populations, \(X_1\) and \(X_2\) are independent and have finite variance.

Then, following the FANOVA decomposition in Equation~\ref{eq:fanova}, the components for each subset satisfy the following for each population:
\[
\mL(f, H_X, \emptyset)(x) = 0, \quad \mL(f, H_X, \{1\})(x) = x_1, \quad \mL(f, H_X, \{2\})(x) = x_2, \quad \mL(f, H_X, \{1,2\})(x) = 0,
\]
\[
\mL(f, K_X, \emptyset)(x) = \mu, \quad \mL(f, K_X, \{1\})(x) = x_1 - \mu, \quad \mL(f, K_X, \{2\})(x) = x_2, \quad \mL(f, K_X, \{1,2\})(x) = 0.
\]

Hence, the difference in means due to differences in \(Y \mid X\) for the component of \(S = \{1\}\) is given by:
\[
\Delta(\mathcal{L}, f, H_X, K_X, \{1\}) = \mathbb{E}_{H_X} \left[ \mL(f, K_X, \{1\})(X) - \mL(f, H_X, \{1\})(X) \right] = \mathbb{E}_{H_X}[X_1 - \mu - X_1] = - \mu \neq 0.
\]

Since this term is not equal to zero, FANOVA misattributes effects to \(Y \mid X\) in this example.
\end{example}

Similarly, the following example demonstrates that ALE misattributes in a simple scenario.

\begin{example}[ALE]
\label{ex:ale}
Let \(f^{K} = f^{H} = f(x) = x_1x_2\), and assume for population H, \(X_1 \sim N(1,1)\), \(X_2 \sim N(0,1)\), and for population K, \(X_1 \sim N(0,1)\), \(X_2 \sim N(\mu,1)\), where \(\mu \neq 0\). 
Assume we observed data \(\{(x_{i,1}^j, x_{i,2}^j)\}_{i=1}^n\), with \(n\) sufficiently large, for \(j = H, K\). 
Following the ALE decomposition in Equation~\ref{eq:ale}, we can compute the centered components for each population:
\[
\mL(f, H_X, \{1\})(x) = 0, \quad
\mL(f, H_X, \{2\})(x) = (x_2 - x_{\min, 2}^H) - (-x_{\min, 2}^H) = x_2,
\]
\[
\mL(f, K_X, \{1\})(x) = \mu (x_1 - x_{\min, 1}^K) - (-\mu x_{\min, 1}^K) = \mu x_1, \quad
\mL(f, K_X, \{2\})(x) = 0.
\]

Thus, the difference in means due to differences in \(Y \mid X\) for \(S = \{1\}\), is given by:
\begin{align*}
    \Delta(\mathcal{L}, f, H_X, K_X, \{1\}) & = \mathbb{E}_{H_X} [\mL(f, K_X, \{1\})(X) - \mL(f, H_X, \{1\})(X)] = \mathbb{E}_{H_X} [\mu X_1 - 0] \\ 
    & = \mu \mathbb{E}_{H_X}[X_1] = \mu \cdot 1 = \mu \neq 0.
\end{align*}

Since this term is not equal to zero, ALE misattributes effects to \(Y \mid X\) in this example.
\end{example}

Examples \ref{ex:fanova} and \ref{ex:ale} show that both ALE and FANOVA can misattribute differences in the covariates to differences in \( Y \mid X \).
With the knowledge that common functional decompositions like FANOVA or ALE \emph{can} misattribute effects, it behooves us to understand how widespread this behavior is.
Does misattribution happen frequently for common functional decompositions? And what properties of functional decompositions will prevent misattribution?
Practitioners need answers to these questions to confidently use methods similar to the one described in Definition~\ref{def:decomp}.
In the next sections, we provide partial answers to these questions.
In Section~\ref{sec:gfan_misattribution}, we argue that FANOVA misattributes in the sense of Definition~\ref{def:decomp} in almost all cases of practical interest, rendering it unsuitable in practice.
And in Section~\ref{sec:mainthm}, we conjecture---and partially prove---that any reasonable functional decomposition that depends on the input covariate distribution will misattribute.

\subsection{When does FANOVA misattribute effects?}
\label{sec:gfan_misattribution}

In the last section, we presented an example in which FANOVA misattributes differences in \(X\) to differences in \(Y \mid X\). 
However, without a precise characterization of when this misattribution occurs, one might think it is specific to the example rather than a general phenomenon. 
We now provide theoretical characterizations of when FANOVA misattributes effects in cases of practical interest. 
Specifically, we first show that FANOVA does not misattribute when the lower-dimensional components computed for population K have a mean of zero under population H. 
We then demonstrate that this condition is highly unrealistic, even in simple cases such as affine functions, and becomes even more restrictive when we allow for more flexibility.



\begin{theorem}[FANOVA attribution]
    \label{thm:expected_0}
    Let \(\mL\) denote the FANOVA decomposition. Then, for any Lebesgue measurable function \( f \), any pair of probability measures \( H_X \) and \( K_X \), and any subset of the covariates \( S \in 2^{[d]} \setminus \{\emptyset\} \), we have 
    \[
        \Delta(\mathcal{L}, f, H_X, K_X, S) = 0
    \] 
    if and only if
    \begin{equation}
    \label{eq:g-favnova-orthogonality}
        \E_{H_X}\Bigl[\mL(f, K_X, S)(X)\Bigr] = 0.
    \end{equation}
    \end{theorem}
\begin{proof}
    By definition, we have \(\Delta(\mathcal{L}, f,H_X,K_X,S) = \E_{H_X}[\mL(f, K_X, S)(X) - \mL(f, H_X, S)(X)]\).
    The mean-zero property of the FANOVA components (see Appendix Lemma~\ref{lemma:orth_hooker2007}) implies that \(\E_{H_X}[\mL(f, H_X, S)(X)] = 0.\)
    Thus, \(\Delta(\mathcal{L}, f,H_X,K_X,S) = \E_{H_X}[\mL(f, K_X, S)(X)]\). It follows immediately that \(\Delta(\mathcal{L}, f,H_X,K_X,S) = 0\) if and only if \(\E_{H_X}[\mL(f, K_X, S)(X)] = 0\).
\end{proof}

Theorem~\ref{thm:expected_0} states that the expectation under population \( H \) of the component computed under population \(K\) must have a mean of zero. 
This suggests a close relationship between the two distributions or that the moments of the marginal distribution must satisfy specific conditions for Equation~\ref{eq:g-favnova-orthogonality} to hold.
If these conditions hold in practical scenarios, then FANOVA could indeed be a viable option. 
Our next set of results indicates that Equation~\ref{eq:g-favnova-orthogonality} unfortunately cannot be reasonably expected to hold in practice.

We start by studying a particularly simple case---when \(f\) is a given affine function.
We show that even in this case, the conditions under which FANOVA does not misattribute are very restrictive. 

\begin{theorem}[FANOVA affine class]
    \label{thm:affine_clas_fanova}
    Let \(X_1, \dots, X_M\) be independent random variables, and let \(H_X\) and \(K_X\) be two probability measures. Consider a function \(f:\mathbb{R}^M \to \mathbb{R}\) given by \(f(x) = \sum_{m=1}^{M} a_m\, b_m(x_m),\)
    where each coefficient \(a_m \in \mathbb{R} \setminus \{0\}\) and each basis function \(b_m: \mathbb{R} \to \mathbb{R}\) is measurable for \(m=1,\dots,M\). Then, we have
    \begin{equation}
        \text{for all } S \in 2^{[M]} \setminus \{\emptyset\}, \quad \E_{H_X}\Bigl[\mL(f, K_X, S)(X)\Bigr] = 0,
    \end{equation}
    if and only if
    \begin{equation}
        \label{eq:equal_means}
        \text{for all } m \in \{1,\dots,M\}, \quad \E_{H_X}\Bigl[b_m(X_m)\Bigr] = \E_{K_X}\Bigl[b_m(X_m)\Bigr].
    \end{equation}
\end{theorem}
    

See Appendix~\ref{Appendix:FANOVA} for the proof. 
Equation~\ref{eq:equal_means} is a fairly restrictive condition, as we expect covariate means to vary between populations; e.g., the proportion of the workforce with high school diplomas will likely not be \emph{identical} between two cities H and K.

We can now see why FANOVA misattributes effects in Example~\ref{ex:fanova}.
There, the function \( f(x) = x_1 + x_2 \) corresponds to the affine class in Theorem~\ref{thm:affine_clas_fanova}, with coefficients \( a_1 = a_2 = 1 \) and basis functions \( b_1(x_1) = x_1 \), \( b_2(x_2) = x_2 \). 
Since the expectation condition in Equation~\ref{eq:equal_means} does not hold (\(\E_{H_X}[X_1] \neq \E_{K_X}[X_1]\)), FANOVA assigns a nonzero difference to \( S = \{1\} \), leading to misattribution.
 
A more revealing version of Theorem~\ref{thm:affine_clas_fanova} would extend to richer basis representations by using multiple basis functions per dimension rather than a single one and allowing for correlated covariates. 
However, we conjecture that this would further restrict the relationship between the moments of the distributions \(H_X\) and \(K_X\), imposing increasingly stringent requirements on them.
In other words, adding flexibility to our model of \(\E[Y \mid X]\) comes at the cost of restricting the set of populations where the decomposition remains accurate. 
This culminates in our next result, which shows that placing minimal restrictions on \(\E[Y \mid X]\) imposes maximal restrictions on the distribution of \(X\).

\begin{theorem}[FANOVA unrestricted]
    \label{thm:unrestricted_fanova}
    Let \(\measurable\) denote the set of measurable functions on the covariate space. Suppose that \(H_X\) and \(K_X\) are probability measures such that \(H_X\) is absolutely continuous with respect to \(K_X\) \((H_X \ll K_X)\). Then, we have
    \[
    \text{for all } f \in \measurable \text{ and for all } S \in 2^{[d]} \setminus \{\emptyset\}, \quad \E_{H_X}\Bigl[\mL(f, K_X, S)(X)\Bigr] = 0,
    \]
    if and only if
    \[
    H_X = K_X, \quad \text{\(K_X\)-almost surely.}
    \]
    \end{theorem}
    
%
See Appendix~\ref{app:fanova_unrestricted} for a proof.

Theorem~\ref{thm:unrestricted_fanova} says that if we want FANOVA to never misattribute for a given pair of distributions---that is, not misattribute for every measurable function \( f \) and every nonempty subset \( S \) of the covariates---then it is necessary and sufficient that the input covariate distributions are identical (i.e., \( H_X = K_X \), up to a \( K_X \)-null set). 
Equivalently, if \( H_X \neq K_X \), then there exists at least one problematic measurable function \( f \) and nonempty subset \( S \) for which FANOVA misattributes to $Y \mid X$. 
In practice, we generally compare distinct populations (i.e., \( H_X \neq K_X \) ), implying that FANOVA will misattribute in settings where \( f \) is one of the problematic cases.
Theorem~\ref{thm:unrestricted_fanova} does not characterize the problematic \( f \), suggesting that knowledge or assumptions about $f$ could rule out misattribution in some applications.
A more practical result would characterize the set of problematic $f$ for a particular set of input densities; we leave this as a direction for future work.
A more concerning result would instead give conditions under which misattribution can occur for any given $f$; we also leave this as a direction for future work.

\section{When do functional decompositions misattribute effects?}
\label{sec:mainthm}

Given the pessimistic results in Section~\ref{sec:failure}, one may be reasonably concerned that \emph{any} functional decomposition $\mL$ may misattribute, making the decomposition of Definition~\ref{def:decomp} of little value, as practitioners would never know when to trust its outputs. 
To resolve this problem, we attempt to characterize what properties of the functional decomposition \(\mL\) cause misattribution. 
We show that under regularity conditions, a functional decomposition \(\mL\) does not misattribute the effects of \(Y \mid X\) if and only if its derivative with respect to the probability measure is orthogonal to \(K\) in the \(L^2\) sense. 
Furthermore, we prove that Definition~\ref{def:decomp} does not lead to misattribution if \(\mL\) is independent of its input distribution. 
For a reverse direction statement, we conjecture that under reasonable assumptions on \(\mL(f,K_X,S)\), the function \(f\), and the distributions \(H_X\) and \(K_X\), a functional decomposition \(\mL\) does not misattribute the effects of \(Y \mid X\) if it does not depend on its input distribution.


We now aim to characterize conditions on the functional \( \mL \) under which \( \Delta(\mathcal{L}, f,H_X,K_X,S) \) does or does not equal zero for all \( f, H_X, K_X \). We start by studying the discrete case and leave the continuous generalization as a conjecture.

\subsection{The discrete setting}
\label{sec:discrete_main}

First, suppose that \(\Xspace\) is a countable space so that the covariates of the random vector \(X = (X_1, \dots, X_d)^T\) are discrete. 
Let \(k(x)\) and \(h(x)\) be the probability mass functions corresponding to the probability measures \(K_X\) and \(H_X\), respectively, with shared support on \(\Xspace\). 
For example, in healthcare, \(\Xspace\) might represent patient categories based on age or pre-existing conditions, while \(f(x)\) could denote the expected recovery time, and \(k(x)\), \(h(x)\) represent the proportions of patients in different hospitals.

Before stating our main theorem of this section, we impose the following regularity conditions on the class of functional decompositions we consider.

\begin{assumptions}
\label{ass:discrete_differentiability}
    The map \(K_X \mapsto \mL(f, K_X, S)\) is twice continuously differentiable with respect to \(K_X\), and its second derivative is uniformly bounded. 
    For a mathematical formulation see Appendix Assumption~\ref{assump:main_discrete}.
\end{assumptions}

For any fixed measurable function \(f\) and subset \(S \in 2^{[d]}\), we denote by \(\mathcal{J}_{K_X}(f, H_X, S)\) the Jacobian matrix of the mapping \(K_X \mapsto \mL(f, K_X, S),\) with respect to \(K_X\), evaluated at \(K_X = H_X\).

We now state our main result for the discrete case, which characterizes when a functional decomposition \(\mL\) will never misattribute the effects of \(Y \mid X\) to changes in \(X\).

\begin{theorem}[Discrete characterization]
    \label{thm:discrete_characterization}
    Under Assumption~\ref{ass:discrete_differentiability} and for all \( H_X, K_X \in \intsim \), we have
    \[
    \Delta(\mathcal{L}, f,H_X,K_X,S) := \delta^{Y \mid X}_S := \mathbb{E}_{H_X} \Bigl[ \mL(f, K_X, S)(X) - \mL(f, H_X, S)(X) \Bigr] = 0,
    \]
    if and only if
    \[
    \text{for all } K_X \in \intsim, \quad \mathcal{J}_{K_X} \mL(f, K_X, S) = \mathbf{c} \, \mathbf{1}^T, \quad \text{for some } \mathbf{c} \in \mathbb{R}^d.
    \]
    \textbf{Remark:} The condition \(\mathcal{J}_{K_X} \mL(f, K_X, S) = \mathbf{c} \, \mathbf{1}^T\) implies that all columns of the Jacobian are identical, so that its rank is \(1\).
\end{theorem}


    

See Appendix~\ref{app:discrete_main} for the proof.

Theorem~\ref{thm:discrete_characterization} shows that if we require a functional decomposition \(\mL\) to never to misattribute the effect of \(Y \mid X\) for \emph{any} distribution \(K_X\), its dependence on \(K_X\) becomes severely restricted. 
Concretely, if the \emph{average} change of \(\mL\) is zero for every pair of distributions \(H_X, K_X \in \intsim\), then all the columns of the Jacobian of \(\mL\) with respect to \(K_X\) must be the same.
This structure means that \(\mL\) cannot distinguish between different probability distributions, which implies that $\Delta$ must be zero. 
As the following corollary shows, this characterization implies \(\mL\) must be constant in its second argument across values in \(\intsim\). 


\begin{corollary}
    \label{cor:constant_M}
    Under the same assumptions as Theorem~\ref{thm:discrete_characterization}, \(\mL\) will not misattribute the effects of \(Y \mid X\) if and only if \(\mL(f,K_X,S) = \mL(f, H_X, S)\) for all \(K_X, H_X \in \intsim\), i.e., the functional \(\mL\) is constant with respect to the distribution over covariates.
\end{corollary}
\begin{proof}
	From Theorem~\ref{thm:discrete_characterization}, there will be no misattribution if and only if \(\mathcal{J}_{K_X}(f, H_X, S) = \mathbf{c1}^T\) for some \(\mathbf{c} \in \mathbb{R}^d\).
    From the Mean Value Theorem, \(\mL(f,K_X,S) - \mL(f,H_X,S) = \mathcal{J}_{K_X} \mL(f,\tilde{H_X},S)(K_X-H_X)\), and since for all \(K_X \in \intsim\), \(\mathcal{J}_{K_X}(f, K_X, S) = \mathbf{c} \mathbf{1}^T\), we have \(\mathbf{c}\mathbf{1}^T(K_X-H_X) = 0\), implying \(\mL(f,K_X,S) = \mL(f,H_X,S)\).
\end{proof}

That is, for a decomposition not to misattribute, \(\mL\) must be constant in \(\intsim\), meaning it is completely unresponsive to changes in \(K_X\). 
We note that this may be unnecessarily restrictive in practice.
In particular, in most applications, we are not concerned with \emph{every} possible redistribution of probability mass but rather with specific structured changes that carry meaningful information. 
Still, as the following example shows, there is at least one existing decomposition that satisfies the conditions of Corollary~\ref{cor:constant_M}.
\begin{example}[Non-generalized FANOVA \citep{hooker2004discovering}]
    \label{ex:non_gen_fanova}
    The non-generalized FANOVA decomposes \(f(x)\) using a uniform distribution \(U\) over the unit hypercube, independent of \(H_X\) or \(K_X\). 
    When \(f^K = f^H = f\), the decomposition remains constant over \(\intsim\), i.e., \(\mL(f, K_X, S) = \mL(f, U, S) = \mL(f, H_X, S)\). Consequently, its Jacobian is zero for all \(K_X \in \intsim\), trivially satisfying \(\mathcal{J}_{K_X} \mL(f, K_X, S) = \mathbf{c} \, \mathbf{1}^T\).
\end{example}
In contrast, we conjecture that \emph{generalized} FANOVA does satisfy the conditions of Corollary~\ref{cor:constant_M}.
\begin{conjecture} \label{cor:fanova_bad}
    Suppose \( f \) is non-constant, and let \( \mL(f, K_X, S) \) be the FANOVA decomposition. Then, \( \mL \) satisfies Assumption~\ref{ass:discrete_differentiability}, and there exist \( H_X, K_X \in \intsim \) such that it misattributes effects of \( Y \mid X \).
\end{conjecture}
Put together, Example~\ref{ex:non_gen_fanova} and Corollary~\ref{cor:fanova_bad} tell a story that is exactly the opposite of typical observations in the functional decomposition literature.
In particular, that generalized FANOVA and ALE depend on their input distributions is tyipcally viewed as beneficial; indeed, this is a major motivator for the work of \citet{apleyVisualizingEffectsPredictor2020} and \citet{hookerGeneralizedFunctionalANOVA2007}.
One reason for this benefit is that typical functions $f$ of interest are often machine learning models fit to training data drawn from covariate distribution $K_X$.
Many flexible machine learning models have arbitrary behavior outside the support of the training data; thus methods like the non-generalized FANOVA that integrate with respect to the uniform distribution may integrate over nonsensical values of $f$.
ALE and FANOVA resolve this issue by integrating over the covariate distribution $K_X$.

These results highlight a tension in the design of functional decomposition methods: non-use of the covariate distribution $K_X$ may result in strange behavior by integrating over nonsensical values of $f$, while use of the covariate distribution may result in nonsensical decompositions of differences between two populations.
We leave as future work an attempt to resolve this seeming contradiction. 

In practice, many applications involve continuous distributions, where densities vary smoothly rather than being restricted to discrete points. 
For example, in economic models, income distributions are continuous, and in healthcare, biomarkers like blood pressure or cholesterol levels are measured on a continuous scale.
To extend our characterization to these cases, we analyze the continuous setting in Appendix~\ref{sec:continuous_main}.

\section{Conclusion}
\label{sec:conclusion}


In this work, we present a natural extension of the Kitagawa-Oaxaca-Blinder decomposition for explaining differences in means to non-linear models of the conditional expectation. 
We note that functional decompositions like FANOVA and ALE seem at first glance like excellent candidates to incorporate into our KOB extension. 
However, we provide simple counterexamples showing that both FANOVA and ALE can incorrectly assign differences in the distribution of covariates \(X\) to differences in the outcome-given-covariates, \(Y \mid X\). 
We further provide characterizations of when FANOVA misattributes for practical cases of interest, as well as a general property that any discrete decomposition should satisfy to never misattribution: the decomposition must be constant across all distributions.
For the general continuous case, we show that if the decomposition is independent of its input distribution, it does not misattribute. 
For a reverse statement, we conjecture that the same will hold as in the discrete case: any reasonable functional decomposition method that depends on its input distribution in a meaningful way will misattribute. 

Our findings highlight a fundamental limitation: methods effective for single-population analysis may be unreliable for comparing differences between populations.
Our work also suggests that the requirements for single-population decomposition and two-population difference decomposition may diverge, highlighting the importance of developing new methods to decompose the difference in means.
Future work should explore how to develop decompositions that avoid misattribution while preserving interpretability in real-world applications.

\section*{Acknowledgments}
The authors are grateful to Raj Agrawal for his contributions and insights during the initial stages of this project. This research was supported in part by an ONR Early Career Grant. Advik Shreekumar was supported in part by the National Science Foundation Graduate Research Fellowship under Grant No.\ 1122374.

DISTRIBUTION STATEMENT A. Approved for public release. Distribution is unlimited.
This material is based upon work supported by the Combatant Commands under Air Force Contract No. FA8702-15-D-0001 or FA8702-25-D-B001. Any opinions, findings, conclusions or recommendations expressed in this material are those of the author(s) and do not necessarily reflect the views of the Combatant Commands.

\bibliography{iclr2025_conference}
\bibliographystyle{iclr2025_conference}

\appendix

\section{FANOVA and ALE misattribution example}
\label{app:examples}

\subsection{Example~\ref{ex:fanova}: FANOVA}
\label{ex:fanovaAppendix}

Consider the setting where the covariates \(X_1\) and \(X_2\) are independent, and let \(f(x) = x_1 + x_2,\) be the model used for both populations. 
Assume that for population \(H\) we have \(\E_{H_X}[X_1] = 0 \text{ and } \E_{H_X}[X_2] = 0\) and for population \(K\) we have \(\E_{K_X}[X_1] = \mu \neq 0 \text{ and } \E_{K_X}[X_2] = 0.\) 
By Lemma~\ref{lemma:fanova_reduction}, under independence the FANOVA decomposition is equivalent to the recursive formula (see Equation~\ref{eq:recursive_formula}) for the Hoeffding–Sobol decomposition \citep{sobol2003theorems, kuo2010decompositions} for a general probability measure. 
Then, the FANOVA components are computed as follows:
\[
\mL(f, H_X, \emptyset)(x) = \E_{H_X}[X_1+X_2] = 0.
\]
\[
\mL(f, H_X, \{1\})(x) = \E_{H_X}[X_1+X_2 \mid X_1=x_1] - \mL(f, H_X, \emptyset)(x) = x_1 + \E_{H_X}[X_2] - 0 = x_1.
\]
\[
\mL(f, H_X, \{2\})(x) = \E_{H_X}[X_1+X_2 \mid X_2=x_2] - \mL(f, H_X, \emptyset)(x) = \E_{H_X}[X_1] + x_2 - 0 = x_2.
\]
\[
\begin{aligned}
\mL(f, H_X, \{1,2\})(x) &= \E_{H_X}[X_1+X_2 \mid X_1=x_1, X_2=x_2] - \mL(f, H_X, \{1\})(x) - \mL(f, H_X, \{2\})(x) - \mL(f, H_X, \emptyset)(x)\\[1mm]
&= (x_1+x_2) - x_1 - x_2 - 0 = 0.
\end{aligned}
\]
Thus, the FANOVA components for population \(H\) are:
\[
\mL(f, H_X, \emptyset)(x)=0,\quad \mL(f, H_X, \{1\})(x)=x_1,\quad \mL(f, H_X, \{2\})(x)=x_2,\quad \mL(f, H_X, \{1,2\})(x)=0.
\]

Similarly, we compute the FANOVA components for population \(K\):
\[
\mL(f, K_X, \emptyset)(x) = \E_{K_X}[X_1+X_2] = \mu.
\]
\[
\mL(f, K_X, \{1\})(x) = \E_{K_X}[X_1+X_2 \mid X_1=x_1] - \mL(f, K_X, \emptyset)(x) = x_1 + \E_{K_X}[X_2] - \mu = x_1 + 0 - \mu = x_1 - \mu.
\]
\[
\mL(f, K_X, \{2\})(x) = \E_{K_X}[X_1+X_2 \mid X_2=x_2] - \mL(f, K_X, \emptyset)(x) = \E_{K_X}[X_1] + x_2 - \mu = \mu + x_2 - \mu = x_2.
\]
\[
\begin{aligned}
\mL(f, K_X, \{1,2\})(x) &= \E_{K_X}[X_1+X_2 \mid X_1=x_1, X_2=x_2] - \mL(f, K_X, \{1\})(x) - \mL(f, K_X, \{2\})(x) - \mL(f, K_X, \emptyset)(x)\\[1mm]
&= (x_1+x_2) - (x_1-\mu) - x_2 - \mu = 0.
\end{aligned}
\]

Thus, the FANOVA components for population \(K\) are:
\[
\mL(f, K_X, \emptyset)(x)=\mu,\quad \mL(f, K_X, \{1\})(x)=x_1-\mu,\quad \mL(f, K_X, \{2\})(x)=x_2,\quad \mL(f, K_X, \{1,2\})(x)=0.
\]

Finally, the difference in the FANOVA effects attributed to \(Y\mid X\) for the subset \(S=\{1\}\) is given by
\[
\Delta(\mathcal{L}, f, H_X, K_X, \{1\}) = \E_{H_X}\Bigl[\mL(f, K_X, \{1\})(X) - \mL(f, H_X, \{1\})(X)\Bigr] = \E_{H_X}[X_1-\mu - X_1] = -\mu \neq 0.
\]

Since this term is nonzero, FANOVA misattributes effects to \(Y\mid X\) in this example.

\subsection{Example~\ref{ex:ale}: ALE}
\label{ex:ALEAppendix}

Recall from Section~\ref{sec:ale} that for a differentiable model 
\(f(x_1,x_2)\) the ALE main effect for \(X_1\) is defined by
\[
    \mL(f, K_X, \{1\})(x) \;=\; \int_{x_{\min,1}}^{x_1} \E_{K_X}\!\Bigl[\frac{\partial f(z_1,X_2)}{\partial z_1}\Bigr]dz_1 \;-\; \text{constant},
\]
with an analogous definition for \(X_2\). Here, the constant is chosen so that the ALE effect has mean zero over the observed data, we will denote these constants by \(c_i^j\) for \(i = \{1,2\}\) and \(j = \{H,K\}\).
Consider the function \(f(x_1,x_2)=x_1x_2,\) which is used for both populations. 
The relevant partial derivatives are \(\frac{\partial f}{\partial x_1}=x_2 \text{ and } \frac{\partial f}{\partial x_2}=x_1.\)

For population \(H\) assume \(X_1\sim N(1,1) \text{ and } X_2\sim N(0,1).\) 
Then, because \(X_1\) and \(X_2\) are independent, \(\E_{H_X}[X_2\mid X_1=z]=\E_{H_X}[X_2]=0.\)
Thus, the ALE component for \(X_1\) is
\[
    \mL(f, H_X, \{1\})(x)
    \;=\; \int_{x_{\min,1}^H}^{x_1} 0\,dz \;-\; c_1^H \;=\; 0 - c_1^H.
\]

In practice, we compute the constants by setting them equal to the empirical mean of the uncentered ALE component \(\mL(f, H_X, \{2\}) - c_2^H\), assuming a large sample size \(n\) so that the Central Limit Theorem provides a good approximation. Theoretically, we take \(c_1^H\) such that \(\E_{H_X}[\mL(f, H_X, \{1\})(X)] = 0\), which gives \(c_1^H = 0\).

Similarly, for \(X_2\), since \(\E_{H_X}[X_1\mid X_2=z]=\E_{H_X}[X_1]=1,\)
we obtain
\[
    \mL(f, H_X, \{2\})(x)
    \;=\; \int_{x_{\min,2}^H}^{x_2} 1\,dz \;-\; c_2^H
    \;=\; (x_2-x_{\min,2}^H) \;-\; c_2^H.
\]

Choosing \(c_2^H\) so that \(\E_{H_X}[\mL(f, H_X, \{2\})(X)]=0\) forces \(c_2^H=-x_{\min,2}^H\), and thus
\[
    \mL(f, H_X, \{2\})(x)=x_2.
\]

Hence, 
\[
    \mL(f,H_X,\{1\})(x)=0 \; \text{ and } \; \mL(f,H_X,\{2\})(x)=x_2.
\]

For population \(K\) assume \(X_1 \sim N(0,1) \text{ and } X_2 \sim N(\mu,1)\) with \(\mu \neq 0.\) 
Then, by independence, \(\E_{K_X}[X_2\mid X_1=z] = \E_{K_X}[X_2] = \mu\) and
\[
    \mL(f, K_X, \{1\})(x)
    \;=\; \int_{x_{\min,1}^K}^{x_1} \mu\,dz \;-\; \text{constant}
    \;=\; \mu(x_1 - x_{\min,1}^K) \;-\; c_1^K,
\]
where the constant \(c_1^K\) solves \(\E_{K_X}[\mL(f, K_X, \{1\})(X)]=0\), thus \(c_1^K = -\mu x_{\min,1}^K\).
Similarly, for \(X_2\), since \(\E_{K_X}[X_1\mid X_2=z] = \E_{K_X}[X_1] = 0,\)
we obtain
\[
    \mL(f, K_X, \{2\})(x)
    \;=\; \int_{x_{\min,2}^K}^{x_2} 0\,dz \;-\; \text{constant}
    \;=\; 0 \;-\; c_2^K,
\]
where \(c_2^K = 0\). Thus,
\[
    \mL(f, K_X, \{1\})(x)
    \;=\; \mu(x_1 - x_{\min,1}^K) - \bigl[-\mu\,x_{\min,1}^K\bigr]
    \;=\; \mu\,x_1 \;
    \text{ and } \;\mL(f, K_X, \{2\})(x)
    \;=\; 0 - 0
    \;=\; 0.
\]

The difference in the ALE effects attributed to \(Y\mid X\) for the change in the covariate \(S=\{1\}\) is given by
\[
    \Delta(\mathcal{L}, f, H, K, \{1\}) 
    \;=\; \E_{H_X}\Bigl[\mL(f, K_X, \{1\})(X) - \mL(f, H_X, \{1\})(X)\Bigr]
    \;=\; \E_{H_X}\bigl[\mu\,X_1 - 0\bigr]
    \;=\; \mu\,\E_{H_X}[X_1]
    \;=\; \mu\cdot 1 
    \;=\; \mu \neq 0.
\]

Since this term is not equal to zero, ALE misattributes effects to \( Y \mid X \) in this example.

\section{FANOVA}
\label{Appendix:FANOVA}

In this section, we formalize the statements from Section~\ref{sec:gfan_misattribution} regarding the characterization of when FANOVA misattributes. 

\subsection{Notation and assumptions}
We assume a general probability measure \( P_X \), which will denote either \( H_X \) or \( K_X \), defined on \((\Xspace, \mathcal{B}(\Xspace))\), where \(\mathcal{B}(\Xspace)\) denotes the Borel sigma algebra.
We assume the measure \( P_X \) is absolutely continuous with respect to a \(\sigma\)-finite reference measure \(\lambda\), with density \( p_X = \frac{dP_X}{d\lambda} \), usually the counting or Lebesgue measure.

The associated Hilbert space is:
\[
    L^2(P_X) = \left\{ g : \mathbb{R}^d \to \mathbb{R} \ \middle| \ \int_{\mathbb{R}^d} g^2(x) p_X(x) \, d\lambda(x) < \infty \right\},
\]
with inner product and norm defined as:
\[
    \langle g_1, g_2 \rangle_{L^2(P_X)} = \int_{\mathbb{R}^d} g_1(x) g_2(x) p_X(x) \, d\lambda(x), \quad \|g\|_{L^2(P_X)} = \sqrt{\langle g, g \rangle_{L^2(P_X)}}.
\]

For the specific cases where \( P_X = H_X \) or \( P_X = K_X \), we denote their densities as \( h(x) = \frac{dH_X}{d\lambda} \) and \( k(x) = \frac{dK_X}{d\lambda} \), respectively. 

We denote by \(\measurable\) the space of \(\lambda\)-measurable functions on the covariates, representing flexible regression models for the conditional expectation of \(Y \mid X\). 
That is,
\[
    \measurable = \{f: \Xspace \to \mathbb{R} \mid f \text{ is \(\lambda\)-measurable}\}.
\]

Note that, since functions in \( \measurable \) are \( \lambda \)-measurable, they are also measurable with respect to probability measures that are absolutely continuous with respect to \(\lambda\).
We use \( \subseteq \) to denote a subset of or equal to a set, and \( \subsetneq \) to indicate a proper subset of a set.  

When working with the Lebesgue measure, we will write \(dx\) in place of \(d\lambda(x)\) for readability, particularly in cases where integrals and densities are defined with respect to \(\lambda\), while keeping \(dx\) for standard notation in functionals and expectations.

\subsection{Proof of Theorem~\ref{thm:affine_clas_fanova}}
The following lemma offers an alternative formulation of the orthogonality conditions in Equation~\ref{eq:HO}, ensuring that all components have zero mean with respect to their corresponding input distributions.

\begin{lemma}[\citetinside{Lemma 1}{hookerGeneralizedFunctionalANOVA2007}]
    \label{lemma:orth_hooker2007}
    The orthogonality conditions in Equation~\ref{eq:HO} hold over \( L^2(\mathbb{R}^d) \) if and only if the integral conditions
    \begin{equation}
    \label{eq:lemma_1Hooker}
        \forall S \subseteq 2^{[d]}, \ \forall i \in S, \quad \int \mL(f,K_X,S)(x) k(x) \, dx_i \, dx_{-S} = 0
    \end{equation}
    are satisfied.
\end{lemma}

Equations~\ref{eq:lemma_1Hooker} are sometimes referred to as the \textit{Weak Annihilating Conditions} \citep{Rahman2014}, and we will refer to these later.

\begin{corollary} 
    \label{cor:meanzero}
    All FANOVA components have mean zero under their input distribution: \(\E_{K_X}[\mL(f,K_X,S)(X)] = 0\).
\end{corollary}
\begin{proof}
    The proof follows directly from Lemma~\ref{lemma:orth_hooker2007} by integrating out the marginal distribution of the covariates not in \(S\).
\end{proof}

The following lemma shows that under the assumption of independence, the FANOVA decomposition is equivalent to the standard result of the Hoeffding-Sobol decomposition \citep{sobol2003theorems, kuo2010decompositions} for a general probability measure.

\begin{lemma}
    \label{lemma:fanova_reduction}
    Let \( X_1, \dots, X_m \) be independent random variables with joint probability distribution \( K_X \). 
    Then, the solution to the FANOVA decomposition problem, as defined in Problem~\ref{eq:fanova}, is equivalent to the following recursive formula:
    \begin{equation}
        \label{eq:recursive_formula}
        \mL(f, K_X, S)(x) =
        \begin{cases}
            \mathbb{E}_{K_X}[f(X)], & \text{if } S = \emptyset, \\[10pt]
            \mathbb{E}_{K_X}[f(X) \mid X_S] - \sum_{V \subsetneq S} \mL(f, K_X, V)(X), & \text{if } S \neq \emptyset.
        \end{cases}        
    \end{equation}
\end{lemma}

\begin{proof}  
    Under the assumption of independence, the Weak Annihilating Conditions in Equation~\ref{eq:lemma_1Hooker} are equivalent to the stronger condition:
\begin{equation}
    \label{eq:stronger_condition}
    \int \mL(f, K_X, S)(x) \, k_i(x_i) \, dx_i = 0, \quad \forall \,i \in S,    
\end{equation}
which follows by the independence assumption by writing \(k(x) = \prod_{j=1}^d k_j(x_j)\) and integrating out the variables not in \(S\). 

Using this result and integrating the additive representation \(f(x) = \sum_{S \in 2^{[d]}} \mL(f, K_X, S)(x)\) against the set \(-S\), we get the recursive formula. 
For any non-empty subset \(S \in 2^{[d]}\), we integrate:
\[
    \E_{K_X}[f(X)\mid X_S] \underset{(indep.)}{=} \int f(x)\, \prod_{i \in -S} k_i(x_i)\, dx_{-S} = \int \left(\sum_{V \in 2^{[d]}} \mL(f, K_X, V)(x) \right) \prod_{i \in -S} k_i(x_i)\, dx_{-S}
\]
\begin{equation}
    \label{eq:reduced_form}
    = \sum_{V \in 2^{[d]}} \int \mL(f, K_X, V)(x)\, \prod_{i \in -S} k_i(x_i)\, dx_{-S}.
\end{equation}

If \(V \cap (-S) = \emptyset\), i.e., \(V \subseteq S\), then \(\mL(f, K_X, V)(x)\) depends only on \(x_S\). 
Consequently,
\[
   \int \mL(f, K_X, V)(x)\, \prod_{i \in -S} k_i(x_i)\, dx_{-S}
   = \mL(f, K_X, V)(x) \int \prod_{i \in -S} k_i(x_i)\, dx_{-S} = \mL(f, K_X, V)(x).
\]

If \(V \cap (-S) \neq \emptyset\) but \(V \not\subseteq S\), then \(\mL(f, K_X, V)(x)\) depends on at least one coordinate in \(-S\). 
Due to Equation~\ref{eq:stronger_condition} and the independence of \(k(x)\), we have:
\[
   \int \mL(f, K_X, V)(x)\, \prod_{i \in -S} k_i(x_i)\, dx_{-S} = 0.
\]
Thus, Equation~\ref{eq:reduced_form} becomes
\[
    \E_{K_X}[f(X)\mid X_S] = \sum_{V \subseteq S} \mL(f, K_X, V)(x) = \sum_{V \subsetneq S} \mL(f, K_X, V)(x) + \mL(f, K_X, S)(x) 
\]
\[
    \iff \mL(f, K_X, S)(x)  = \E_{K_X}[f(X)\mid X_S]  - \sum_{V \subsetneq S} \mL(f, K_X, V)(x).
\]

Similarly, for \(S = \emptyset\), taking the expectation over all variables \(X\) yields the corresponding formula: \(\mL(f, K_X, \emptyset) = \E_{K_X}[f(X)]\).

\end{proof}

Lemma~\ref{lemma:fanova_reduction} was previously shown by \citet[Corollary 4.6]{Rahman2014_2} in the context of the generalized ANOVA dimensional decomposition, which reduces to the standard FANOVA dimensional decomposition. 
We restate it here using the specific terminology of FANOVA and provide a concrete proof to derive the reduced formula solution.

We can now formally prove Theorem~\ref{thm:affine_clas_fanova} which characterizes the conditions under which FANOVA does not misattribute for a given affine function of the covariates \(f\).

\begin{proof}[Proof of Theorem~\ref{thm:affine_clas_fanova}]
    \((\Rightarrow)\) Suppose that 
    \[
    \mathbb{E}_{H_X} \left[ \mL(f, K_X, S)(X) \right] = 0, \quad \text{for all } S \in 2^{[d]} \setminus \{\emptyset\}.
    \]

    By Lemma~\ref{lemma:fanova_reduction} we have that under the independence assumption, FANOVA reduces to the recursive form solution in Equation~\ref{eq:recursive_formula}.
    Then, for \( S \neq \emptyset \), the component would be:
    \[
    \mL(f, K_X, \emptyset)(x) = \mathbb{E}_{K_X}[f(X)] = \sum_{m=1}^M a_m \mathbb{E}_{K_X}[b_m(X_m)].
    \]

    For the single effects, \( S = \{m\} \), the additive components would take the form:
    \[
    \mL(f, K_X, \{m\})(x) = \mathbb{E}_{K_X}[f(X) | X_m] - \mL(f, K_X, \emptyset)(x).
    \]

    Expanding this, and using the independence assumption, we get:
    \[
    \begin{aligned}
        \mL(f, K_X, \{m\})(x) &= \left( a_m b_m(x_m) 
        + \sum_{j \neq m} a_j \mathbb{E}_{K_X}[b_j(X_j)] \right) - \left( \sum_{j=1}^M a_j \mathbb{E}_{K_X}[b_j(X_j)] \right) \\
        &= a_m \left( b_m(x_m) - \mathbb{E}_{K_X}[b_m(X_m)] \right), 
        \quad \forall \, m = 1, \dots, M.
    \end{aligned}
    \]

    Lastly, for \( S \) with \( |S| \geq 2 \), we have that the recursive formula for the Functional ANOVA components is given by
\[
    \mL(f, K_X, S)(x) = \mathbb{E}_{K_X}[f(X) \mid X_S] - \sum_{V \subsetneq S} \mL(f, K_X, V)(x).
\]

Where,  \( \mathbb{E}_{K_X}[f(X) \mid X_S] = \sum_{m \in S} a_m b_m(x_m) + \sum_{m \notin S} a_m \mathbb{E}_{K_X}[b_m(X_m)]\) and the sum of lower-order terms takes the form:

\begin{equation*}
    \begin{aligned}
        \sum_{V \subsetneq S} \mL(f, K_X, V)(x) & = \mL(f, K_X, \emptyset) + \sum_{m \in S} \mL(f, K_X, \{m\})(x) \\
        & =  \sum_{m=1}^M a_m \mathbb{E}_{K_X}[b_m(X_m)] + \sum_{m \in S} a_m \bigl(b_m(x_m) - \mathbb{E}_{K_X}[b_m(X_m)]\bigr) \\
        & = \sum_{m \in S} a_m b_m(x_m) + \sum_{m \notin S} a_m \mathbb{E}_{K_X}[b_m(X_m)].
    \end{aligned}
\end{equation*}

Substituting into the recursive formula:
\[
\mL(f, K_X, S)(x) = \Bigl(\sum_{m \in S} a_m b_m(x_m) + \sum_{m \notin S} a_m \mathbb{E}_{K_X}[b_m(X_m)]\Bigr) - \Bigl(\sum_{m \in S} a_m b_m(x_m) + \sum_{m \notin S} a_m \mathbb{E}_{K_X}[b_m(X_m)]\Bigr) = 0.
\]

That is,
\[
\mL(f, K_X, S)(x) = 0, \quad \forall \, S \text{ with } |S| \geq 2.
\]

Now, by hypothesis, for every \( m = 1, \dots, M \), we have:
\[
    0 = \mathbb{E}_{H_X} \left[ \mL(f, K_X, \{m\})(X) \right] = \mathbb{E}_{H_X} \left[ a_m \left( b_m(X_m) - \mathbb{E}_{K_X}[b_m(X_m)] \right) \right] \iff
    0 = a_m \left( \mathbb{E}_{H_X}[b_m(X_m)] - \mathbb{E}_{K_X}[b_m(X_m)] \right).
\]

We also assumed that \( a_m \neq 0 \) for all \( m = 1, \dots, M \). Therefore, we conclude that
\[
    \mathbb{E}_{H_X}[b_m(X_m)] = \mathbb{E}_{K_X}[b_m(X_m)], \quad \text{for all } m = 1, \dots, M.
\]

\((\Leftarrow)\) Suppose that \(\mathbb{E}_{H_X}[b_m(X_m)] = \mathbb{E}_{K_X}[b_m(X_m)]\) for each  \( m = 1, \dots, M \). Then,
\[
\mathbb{E}_{H_X} \left[ \mL(f, K_X, \{m\})(X_m) \right]
= \mathbb{E}_{H_X} \left[ a_m \left( b_m(X_m) - \mathbb{E}_{K_X}[b_m(X_m)] \right) \right]
= a_m \left( \mathbb{E}_{H_X}[b_m(X_m)] - \mathbb{E}_{K_X}[b_m(X_m)] \right) = 0.
\]
\end{proof}

\subsection{Proof of Theorem~\ref{thm:unrestricted_fanova}}
\label{app:fanova_unrestricted}

Note that in Theorem~\ref{thm:unrestricted_fanova} we have assumed that any measurable function of the covariates satisfies an additive decomposition. 
However, this assumption can be relaxed by imposing a couple of conditions, such as the probability measure being dominated by a product measure \citep[Equation C.1]{pmlr-v22-hoffman12}, and a boundedness assumption on the densities \citep[Equations C.2 or C.3]{pmlr-v22-hoffman12}. 
Under these assumptions, it follows that for any square-integrable measurable function of the covariates, there exist functions such that the original function admits an additive representation; see \citet[Theorem 1]{pmlr-v22-hoffman12}.

To show Theorem~\ref{thm:unrestricted_fanova}, we first state the definition of mean-zero functions and then prove Lemma~\ref{lemma:mean_zero_prob_measure} and Lemma~\ref{lemma:span_equals_meanzero}, which will serve as intermediate steps for the main proof.

\begin{definition}[Mean-zero square integrable functions]
    \label{def:zeromeansqint}
    We denote by \( W(K_X) \) the space of \textit{mean-zero functions} in \( L^2(K_X) \) with respect to the probability measure \( K_X \):
    \[ 
    W(K_X) = \left\{ \phi \in L^2(K_X) : \int_{\Xspace} \phi(x) k(x) \, d\lambda(x) = 0 \right\}.
    \]
\end{definition}

The following lemma states that the orthogonal complement of the space of mean-zero square-integrable functions with respect to a probability measure is the space of almost surely constant functions.

\begin{lemma}
    \label{lemma:mean_zero_prob_measure}
    Let \( K_X \) be a probability measure on \(\mathbb{R}^d\) with density \( k(x) \) with respect to a reference measure \(\lambda\). 
    The orthogonal complement of \( W(K_X) \) in \( L^2(K_X) \) is the space of constant functions, that is:
    \[
    W(K_X)^\perp = \left\{ f \in L^2(K_X) : f(x) = c,\; K_X\text{-a.s.} \right\},
    \]
    where \( c \in \mathbb{R} \) is a constant.
\end{lemma}

\begin{proof}
    Let
    \[
        V = \left\{ f \in L^2(K_X) : f(x) = c,\; K_X\text{-a.s.} \right\}.
    \]

    We will prove that \( W(K_X)^{\perp} = V \) by first showing that \( V \subseteq W(K_X)^{\perp} \). 
    Let \( f \in V \); then \( f(x) = c \) for \( K_X \)-almost every \( x \in \Xspace\). 
    Thus, for any \( \phi \in W(K_X) \), we have
    \[
    \E_{K_X}[f(X)\phi(X)] = c \cdot \E_{K_X}[\phi(X)] = 0.
    \]

    It remains to show that \( W(K_X)^{\perp} \subseteq V\). 
    Let \( f \in W(K_X)^{\perp} \); then \(\E_{K_X}[f(X)\phi(X)] = 0 \) for every \( \phi \in W(K_X) \). 
    In particular, take an arbitrary \(K_X\)-measurable set \( A \subseteq \Xspace \) and define \(\psi_A(x) = \chi_A(x) - K_X(A)\). 
    Where \( \chi_A(x) \) is the indicator function over \( A \) and \( K_X(A) \) is the probability of the event \(A\) under the distribution \(K_X\). 
    Clearly, \(\psi_A\) belongs to \(W(K_X)\).
    Moreover, 
    \begin{align*}
        0 = \E_{H_X}[f(X)\psi(X)] &= \E_{H_X}[f(X)\chi_A(X)] - \E_{H_X}[f(X)K_X(A)]  \\
        & = \int_A f(x)\,dK_X - K_X(A)\int_{\Xspace} f(x)\,dK_X.
    \end{align*}
    \begin{equation}
        \label{eq:new_measure}
        \iff \int_A f(X)\,dK_X= K_X(A)\int_{\Xspace}f(x)\,dK_X.
    \end{equation}

    Define \(\mu(A) = \int_A f(x)\,dK_X(x)\), which is a signed measure with respect to \(K_X\). 
    By the Radon–Nikodym theorem for signed measures, \(f\) is the unique Radon–Nikodym derivative \(d\mu/dK_X\). 
    Since we also have, by Equation~\ref{eq:new_measure} that 
    \[
        \mu(A) = K_X(A)\int_{\Xspace} f(x)\,dK_X = K_X(A) \cdot c = c \int_A\,dK_X = \int_A c\,dK_X,
    \]
    it follows by the uniqueness of the Radon–Nikodym derivative that \(f(x) = c\) \(K_X\)-almost surely. Therefore, \(V = W(K_X)^{\perp}\).
\end{proof}

The following lemma shows that the space of mean-zero functions is equivalent to the span of hierarchical orthogonal functional ANOVA components, obtained by varying the covariate functions over the entire space of measurable functions.

\begin{lemma}
    \label{lemma:span_equals_meanzero}
    Suppose \( f \in \measurable \) and let \(\mL(f,K_X,S)\) denote the functional ANOVA component corresponding to the subset of covariates \( S \). 
    Define \(\mathcal{F} = \{ \mL(f, K_X, S)(x) : S \in 2^{[d]}, S \neq \emptyset,\; f \in L^2(K_X) \}\). 
    Then, \(\text{span} \{\mathcal{F}\} = W(K_X)\).
\end{lemma}
    
\begin{proof}
    \((\subseteq)\) Let \(\phi(x) = \sum_{i=1}^n c_i \, \mL(f_i, K_X, S_i)(x_{S_i})\), where \( c_i \in \mathbb{R} \), \( f_i \in L^2(K_X) \), and \( S_i \in 2^{[d]} \setminus \{\emptyset\} \).
    
    Since \(\mathbb{E}_{K_X} \left[ \mL(f_i, K_X, S_i)(X_{S_i}) \right] = 0\), we have \(\mathbb{E}_{K_X}[\phi(X)] = 0\). 
    Moreover, \(\phi\) is square integrable since each summand \(\mL(f_i, K_X, S_i)(x_{S_i})\) is square integrable. Thus, \(\phi \in W(K_X)\) and \(\text{span} \{\mathcal{F}\} \subseteq W(K_X)\).

 \((\supseteq)\) 
 Take an arbitrary $\phi \in W(K_X)$.
 Because $W(K_X) \subseteq L^2(K_X)$, we have that $\phi$ is square integrable and so has an additive decomposition:
 \begin{equation*}
 	\phi(x) = \sum_{S \in 2^{[d]}} \mL(\phi, K_X, S)(x) = \sum_{S \in 2^{[d]} \setminus \{\emptyset\}} \mL(\phi, K_X, S)(x) + \mathbb{E}_{K_X}[\phi(X)] = \sum_{S \in 2^{[d]} \setminus \{\emptyset\}} \mL(\phi, K_X, S)(x),
 \end{equation*}
 where the last equality holds because $E_{K_X}[\phi(X)] = 0$.
 Each term on the right-hand side of the last expression belongs to \(\mathcal{F}\). 
 Thus, \(\phi \in \text{span} \{\mathcal{F}\}\).
\end{proof}
    
With these two results, we now proceed to prove Theorem~\ref{thm:unrestricted_fanova}.

\begin{proof}[Proof of Theorem~\ref{thm:unrestricted_fanova}]
    \((\Rightarrow)\) Suppose \(\mathbb{E}_{H_X} \left[ \mL(f, K_X, S)(X) \right] = 0\). By the mean zero property of FANOVA, \(\mathbb{E}_{K_X} \left[ \mL(f, K_X, S)(X) \right] = 0\). Then,
\begin{align}
    \mathbb{E}_{H_X} \left[ \mL(f, K_X, S)(X) \right]  &= \mathbb{E}_{K_X} \left[ \mL(f, K_X, S)(X) \right] \nonumber \\
    \iff 0 &= \mathbb{E}_{K_X} \left[ \mL(f, K_X, S)(X) \right] - \mathbb{E}_{H_X} \left[ \mL(f, K_X, S)(X) \right] \nonumber \\
    \iff 0 &= \mathbb{E}_{K_X} \left[ \mL(f, K_X, S)(X) \right] - \mathbb{E}_{K_X} \left[ \mL(f, K_X, S)(X) \frac{h(x)}{k(x)} \right] \nonumber \\
    \iff 0& = \int \mL(f, K_X, S)(X) \left( 1 - \frac{h(x)}{k(x)} \right) dK_X, \; \text{ for all } f \in \measurable\text{ and } \text{ for all } S \neq \emptyset. \label{eq:forall_f_s}
\end{align}
    
By Lemma~\ref{lemma:span_equals_meanzero}, as we vary \(f\) and \(S\) over the space \(\measurable \times \bigl(2^{[d]} \setminus \{\emptyset\}\bigr)\), the components \(\{\mL(f, K_X, S)\}_{(f,S)}\) span the mean-zero space; that is, \(\text{span} \{ \mathcal{F} \} = W(K_X)\). 
Therefore, Equation~\ref{eq:forall_f_s} is equivalent to 
\begin{equation*}
    \begin{aligned}
        0 = \int \mL(f, K_X, S)(x) \left( 1 - \frac{h(x)}{k(x)} \right) \,dK_X = 0, \; \text{ for all } \mL(f, K_X, S)(x) \in W(K_X).
    \end{aligned}
\end{equation*}

Thus, \(\left( 1 - \frac{h(x)}{k(x)} \right) \in W(K_X)^\perp\). i.e., \((1 - \frac{h(x)}{k(x)})\) is orthogonal to all zero-mean functions in \( L^2(K_X) \). 
Furthermore, by Lemma~\ref{lemma:mean_zero_prob_measure} we know that the orthogonal space to \( W(K_X) \) is the space of \(K_X\)-almost surely constant functions. 
Thus, there must exist a constant \( c \in \mathbb{R} \) such that \( 1 - \frac{h(x)}{k(x)} = c, \, K_X\text{-a.s}\). 
Finally, noting that \( \int dK_X - \int \frac{h(x)}{k(x)}\,dK_X = \int c\,dK_X\), we have \(c = 0\), \(K_X\text{-a.s.}\) Therefore,
\[
\frac{h(x)}{k(x)} = 1 \quad \Rightarrow \quad h(x) = k(x), \quad K_X\text{-a.s.} \quad \Rightarrow \quad H_X = K_X, \quad K_X\text{-a.s.}
\]

The only if part \((\Leftarrow)\) follows by the mean-zero property of \(\mL(f, K_X, S)(x)\) under \(K_X\).
\end{proof}

\section{Mathematical framework and Section~\ref{sec:mainthm} results}
\label{App:Mathematical framework}

In this section, we describe in detail the mathematical background necessary for Section~\ref{sec:mainthm}, along with the additional notation and assumptions required to prove our main results and state our main conjecture.

\subsection{The discrete case}
\label{app:discrete_main}
We now formalize the assumptions discussed in Section~\ref{sec:discrete_main} and formally prove Theorem~\ref{thm:discrete_characterization}. 
First, let 
\[
\Sigma = \left\{ z \in \mathbb{R}^d : z_i \geq 0 \; \forall i, \; \mathbf{1}^T z = 1 \right\},
\] 

denote the standard \(d-1\) simplex and let \(\intsim\) denote its interior.
Following Definition~\ref{def:zeromeansqint} for the space of mean-zero square integrable functions of the covariates, we denote by \(W\) the space of mean-zero vectors: 
\[
W = \left\{ \phi \in \mathbb{R}^d :  \mathbf{1}^T \phi  = 0 \right\}.
\]

We make the following mild regularity conditions on the functional decomposition \(\mL\). 

\begin{customassumption}
    \label{assump:main_discrete}
    The following hold:
    \begin{enumerate}
        \item \textbf{Twice differentiable}. For any \(f\) and \(S\), the map \(K_X \to \mL(f, K_X, S)\) is twice continuously differentiable.
        \item \textbf{Uniformly bounded condition}. 
        There exists a constant \(M > 0\) such that \(\sup_{(f,K_X,S)} \|\mathcal{H}_{K_X}\mL(f, K_X, S)\|_{\mathrm{op}} \leq M\), where \(\mathcal{H}_{K_X} \mL(f, K_X, S)\) is the Hessian of \(\mL(f, K_X, S)\) with respect to \(K_X\).
    \end{enumerate}
\end{customassumption}

Before proving Theorem~\ref{thm:discrete_characterization}, we show two auxiliary Lemmas: Lemma~\ref{lemma:y_varies} states that for any vector in the interior of the simplex and any mean-zero vector, there always exists a small perturbation along the mean-zero vector that keeps the perturbed vector within the simplex.
Lemma~\ref{lemma:aux_zero_vector} serves as an intermediate step in proving Proposition~\ref{prop:linear_algebra}, which characterizes matrices satisfying a specific condition that the Jacobian of a general functional decomposition must satisfy.

\begin{lemma}
    \label{lemma:y_varies}
    For every \( x \in \intsim \) and \( \phi \in W \), there exists an open interval \( I \subsetneq \mathbb{R} \) containing \( 0 \) such that for all \( \varepsilon \in I \), 
    \(x + \varepsilon \phi \in \intsim\).
\end{lemma}
\begin{proof}
    Let \( y = x + \varepsilon \phi \). 
    We need to show that \( y \in \intsim \) for a suitable choice of \(\varepsilon\) within some interval \(I\). Note that 
    \(\mathbf{1}^\top y = \mathbf{1}^\top x + \varepsilon \mathbf{1}^\top \phi = 1,\)
    which implies that \(y\) satisfies the constraint \(\mathbf{1}^\top y = 1\). Thus, it remains to verify that \(y_i > 0\) for all \(i\).

    If \(\phi_i > 0\), then \(x_i + \varepsilon \phi_i > 0\) holds for all \(\varepsilon\) such that \(\varepsilon > -\frac{x_i}{\phi_i},\) which is also satisfied by taking
    \[
        \varepsilon > \max_{i : \phi_i > 0} \left\{ -\frac{x_i}{\phi_i} \right\}.
    \]

    If \(\phi_i < 0\), then \(x_i + \varepsilon \phi_i > 0\) holds for all \(\varepsilon\) such that \(\varepsilon < -\frac{x_i}{\phi_i},\) which is also satisfied by taking
    \[
        \varepsilon < \min_{i : \phi_i < 0} \left\{ -\frac{x_i}{\phi_i} \right\}, \quad \text{note } -\frac{x_i}{\phi_i} > 0.
    \]

    If \(\phi_i = 0\), then any \(\varepsilon \in \mathbb{R}\) satisfies \(y_i > 0\). Therefore, by choosing
    \[
        \varepsilon \in \left( \max_{i : \phi_i > 0} \left\{ -\frac{x_i}{\phi_i} \right\}, \min_{i : \phi_i < 0} \left\{ -\frac{x_i}{\phi_i} \right\} \right) =: I,
    \]
    we ensure that \(y_i > 0\) for all \(i\). Thus, \(y \in \intsim\).
\end{proof}

\begin{remark}
    Lemma~\ref{lemma:y_varies} is equivalent to stating that for a given \(x \in \intsim\) there exists some point \(y \in \intsim\) and \(\varepsilon > 0\) such that 
    \(\phi = \frac{1}{\varepsilon}(y - x)\).
    That is, we can recover any vector \(\phi \in W\) given an initial \(x\) and a suitable pair \((\varepsilon, y) \in \mathbb{R}_{++} \times \intsim\).
\end{remark}

\begin{lemma}
    \label{lemma:orth_compl}
    \(W^\perp = \text{span}\{\mathbf{1}\}\).
\end{lemma}
\begin{proof}
    \((\supseteq)\)
    Let $u = \alpha \mathbf{1}$ for some scalar $\alpha \in \mathbb{R}$. For any $w \in W$, we have $\mathbf{1}^\top w = 0$. Then:
    \[
        u^\top w = (\alpha \mathbf{1})^\top w = \alpha (\mathbf{1}^\top w) = \alpha \cdot 0 = 0.
    \]

    Thus, \(\mathrm{span}\{\mathbf{1}\} \subseteq W^\perp.\)

    \((\subseteq)\) Let $u \in W^\perp$, that is \(u^\top w = 0\) for all \(w \in W\). 
    Consider the vectors $e_j - e_k$ (where $e_j$ is the $j$-th standard basis vector) with $j \neq k$. 
    Note that: \(\mathbf{1}^\top (e_j - e_k) = 1 - 1 = 0.\)
    Hence, $e_j - e_k \in W$. Since $u \in W^\perp$, it follows that:
    \[
    u^\top (e_j - e_k) = 0 \implies u_j - u_k = 0 \implies u_j = u_k.
    \]

    As \(j\) and \(k\) were arbitrary, all coordinates of \(u\) are equal. 
    Thus, there exists a scalar \(\alpha\) such that $u = \alpha \mathbf{1}$. 
    This implies \(u \in \mathrm{span}\{\mathbf{1}\}\). 
    Therefore, \(W^\perp = \mathrm{span}\{\mathbf{1}\}\).

\end{proof}

\begin{lemma}
    \label{lemma:aux_zero_vector}
    Let \(w \in \mathbb{R}^d\). If
    \(
        x^\top w = 0, \text{ for all } x \in \intsim, \text{ then } w = 0.
    \)
\end{lemma}

\begin{proof}
    We proceed by contradiction. 
    Suppose \(w \neq 0\); we will show that there exists some vector \(x \in \intsim\) such that \(x^\top w \neq 0\).

    Assume there is a component \(i\) such that \(w_i > 0\). Take \(x_i = 1 - \alpha\) and \(x_{j \neq i} = \frac{\alpha}{d-1}\) for some \(\alpha \in (0, 1)\). Clearly, \(x \in \intsim\). We will show that, for a valid \(\alpha\), \(x^\top w > 0\). We start by noting that:
    \begin{align*}
        x^\top w &= (1 - \alpha) w_i + \sum_{j \neq i} \frac{\alpha}{d-1} w_j > 0 
        \iff w_i - \alpha w_i + \frac{\alpha}{d-1} \left( \sum_{j \neq i} w_j \right) > 0 \\
        &\iff w_i + \alpha \left( \frac{\sum_{j \neq i} w_j}{d-1} - w_i \right) > 0.
    \end{align*}
    
    Let \(M = \frac{\sum_{j \neq i} w_j}{d-1} - w_i\). Then, we have the following subcases:

    \textbf{Case 1:} If \(M > 0\), then \(w_i + \alpha M > w_i > 0\) for all \(\alpha > 0\).

    \textbf{Case 2:} If \(M < 0\), then \(\alpha < \frac{-w_i}{M}\), where \(\frac{-w_i}{M} > 0\). Thus, any \(\alpha \in \left(0, \frac{-w_i}{M}\right)\) would imply that \(x^\top w > 0\).

    The case where \(w_i < 0\) for some \(i\) follows by a completely analogous argument. Therefore, if \(w \neq 0\), we can always find an \(x \in \intsim\) such that \(x^\top w \neq 0\).

    Hence, it must be that \(w = 0\).
\end{proof}
    
\begin{proposition}
    \label{prop:linear_algebra}
    Let \(x, y \in \intsim\) and \( A \in \mathbb{R}^{d \times d}\). Then,
    \[
       x^\top A (y - x) = 0, \, \text{ for all } x, y \in \intsim,
    \]
    if and only if
    \begin{equation}
    \label{eq:col_in_span2}
    A = \mathbf{c} \mathbf{1}^\top, \text{ for some } \mathbf{c}  \in \mathbb{R}^d.
    \end{equation}
\end{proposition}

\begin{proof}
    Consider a fixed \(x \in \intsim\). 
    By Lemma~\ref{lemma:y_varies}, for any \(\phi \in W\), there exists \(I \subsetneq \mathbb{R}\) such that \(\forall \varepsilon \in I\), \(y = x + \varepsilon \phi \in \intsim\). 
    Thus, \(\frac{y - x}{\varepsilon} = \phi \in W\), and as we vary \((\varepsilon, y) \in I \times \intsim\), we recover any \(\phi \in W\) (see the proof of Lemma~\ref{lemma:y_varies}). 
    
    Since \(x \neq y\), we have \(\varepsilon \neq 0\). Therefore,
    \[
        0 = x^\top A (y - x) = x^\top A \varepsilon \phi \implies x^\top A \phi=0, \, \forall \phi \in W.
    \]
    
    It follows that \(x^\top A \in W^\perp\) for each \(x \in \intsim\). By Lemma~\ref{lemma:orth_compl}, we know that \(W^\perp = \text{span} \{\mathbf{1}\}\). 
    We claim that if 
    \[
        x^\top A \in \text{Span} \{\mathbf{1}\} \, \text{ for all } x \in \intsim, \text{ then }  A = \mathbf{c}\mathbf{1}^\top, \text{ for some } \mathbf{c} \in \mathbb{R}^d.
    \]

    Suppose \(x^\top A \in \text{Span} \{\mathbf{1}\}\). This implies there exists \(c(x) \in \mathbb{R}\) such that \(x^\top A = c(x) \mathbf{1}^\top\). Let \(a_j\) denote the \(j\)-th column of \(A\). 
    Then, \(x^\top a_j = c(x)\) for all \(j \in \{1, \dots, d\}\). In particular, for \(i \neq j\), we have 
    \begin{equation}
        \label{eq:use_lemma_zero_vector}
        x^\top a_i = x^\top a_j \iff x^\top (a_i - a_j) = 0, \, \text{for all } x \in \intsim.
    \end{equation}
        
    Applying Lemma~\ref{lemma:aux_zero_vector} to the vector \(w = a_i - a_j \in \mathbb{R}^d\), we know that Equation~\ref{eq:use_lemma_zero_vector} implies \(w = a_i - a_j = 0\). 
    Thus, \(a_1 = a_2 = \dots = a_d\), which means \(A = \mathbf{c}  \mathbf{1}^\top\) for some \(\mathbf{c} \in \mathbb{R}^d\).
\end{proof}

Now we proceed to prove our main discrete characterization result.

\begin{proof}[Proof of Theorem~\ref{thm:discrete_characterization}]  
    \((\Rightarrow)\)
    Since \(\mL(f,K_X,S)\) is continuously differentiable, by the Mean Value Theorem there exists \(\tilde{H_X}\) on the line segment between \(K_X\) and \(H_X\), such that
    \begin{equation}
    \label{expansion:mvt}
        \mL(f, K_X, S) - \mL(f, H_X, S) = \mathcal{J}_{K_X}\mL(f, \tilde{H_X}, S)(K_X-H_X),
    \end{equation}
    where \(\mathcal{J}_{K_X}\mL(f, \tilde{H_X}, S)\) is the Jacobian of \(\mL(f,K_X,S)\) with respect to \(K_X\).
    
    Consider the path \( \gamma(t) = H_X + t\,(\tilde{H_X} - H_X), \, t \in [0,1]\). Then \(\gamma(0) = H_X\) and \(\gamma(1) = \tilde{H_X}\). By the fundamental theorem of calculus (in vector form), we have
    \[
      \mathcal{J}_{K_X} \mL(f, \tilde{H_X}, S)
      - \mathcal{J}_{K_X} \mL(f, H_X, S)
      = \int_{0}^{1}
        \mathcal{H}_{K_X} \mL\bigl(f, \gamma(t), S\bigr)\,\bigl[\tilde{H_X} - H_X\bigr]
      \,dt.
    \]
    
    Taking the operator norm on both sides, we obtain
    \[
      \begin{aligned}
      \bigl\|\mathcal{J}_{K_X} \mL\bigl(f, \tilde{H_X}, S\bigr)
            - \mathcal{J}_{K_X} \mL\bigl(f, H_X, S\bigr)\bigr\|
      &= \Bigl\|
          \int_{0}^{1}
            \mathcal{H}_{K_X} \mL\bigl(f, \gamma(t), S\bigr)\,(\tilde{H_X} - H_X)\,dt
        \Bigr\|
      \\
      &\le
        \int_{0}^{1}
          \bigl\|\mathcal{H}_{K_X} \mL\bigl(f, \gamma(t), S\bigr)\bigr\|
          \,\|\tilde{H_X} - H_X\|
        \,dt
      \\
      &\le
        \int_{0}^{1}
          M \,\|\tilde{H_X} - H_X\|
        \,dt,
        \quad (\text{by the uniform bound } \leq M)
      \\
      &=
        M\,\|\tilde{H_X} - H_X\|
        \int_{0}^{1} dt
      \;=\;
        M\,\|\tilde{H_X} - H_X\|.
      \end{aligned}
    \]
    
    Thus we have
    \[
      \bigl\|\mathcal{J}_{K_X} \mL(f, \tilde{H_X}, S)
            - \mathcal{J}_{K_X} \mL(f, H_X, S)\bigr\|
      \;\le\;
      M\,\|\tilde{H_X} - H_X\|.
    \]

    That is,
    \[
      \mathcal{J}_{K_X} \mL(f, \tilde{H_X}, S)
      = \mathcal{J}_{K_X} \mL(f, H_X, S)
      + O\bigl(\|\tilde{H_X} - H_X\|\bigr).
    \]
    
    Substituting back into Equation~\ref{expansion:mvt}, we get
    \begin{equation}
        \label{eq:taylor_exp_order}
        \mL(f, K_X, S) - \mL(f, H_X, S) = \mathcal{J}_{K_X} \mL(f, H_X, S)(K_X-H_X) +  O\bigl(\|K_X - H_X\|^2\bigr).
    \end{equation}

    By Lemma~\ref{lemma:y_varies}, for a fixed \(h\) and for any \(\phi \in W\), there exists \(I \subsetneq \mathbb{R}\) such that \(\forall \varepsilon \in I\), \(k = h + \varepsilon \phi \in \intsim\). 
    Since \(h \neq k\) we have that \(\varepsilon \neq 0\). 
    Thus, \(\frac{y - x}{\varepsilon} = \phi \in W\), and as we vary \((\varepsilon, y) \in I \times \intsim\), we recover any \(\phi \in W\), see proof of Lemma~\ref{lemma:y_varies}.
    Substituting in Equation~\ref{eq:taylor_exp_order} we have that
    \[
    \mL(f, K_X, S) - \mL(f, K_X + \varepsilon\phi, S)  =  \mathcal{J}_{K_X} \mL(f, K_X + \varepsilon\phi, S)(\varepsilon\phi) +  O\bigl(\varepsilon^2\|\phi\|^2\bigr)
    \]
\[
\begin{aligned}
    & \iff \lim_{\varepsilon \to 0} \frac{\mL(f, K_X, S) - \mL(f, K_X + \varepsilon\phi, S)}{\varepsilon} = \lim_{\varepsilon \to 0} \left\{ \mathcal{J}_{K_X} \mL(f, K_X + \varepsilon\phi, S) \phi + O\bigl(\varepsilon\|\phi\|^2\bigr)\right\} \\
    &\underset{(1)}{\iff} \lim_{\varepsilon \to 0} \frac{\mL(f, K_X, S) - \mL(f, K_X + \varepsilon \phi, S)}{\varepsilon} 
    = \lim_{\varepsilon \to 0} \left\{ \mathcal{J}_{K_X} \mL(f, K_X + \varepsilon \phi, S) \phi 
    + O\bigl(\varepsilon \|\phi\|^2\bigr) \right\} \\[1em]
    &\underset{(2)}{\iff} \lim_{\varepsilon \to 0} \frac{\mL(f, K_X, S) - \mL(f, K_X + \varepsilon \phi, S)}{\varepsilon} 
    = \mathcal{J}_{K_X} \mL(f, K_X, S) \phi \\[1em]
    &\underset{(3)}{\iff} \mathbb{E}_{H_X} \left[ \lim_{\varepsilon \to 0} 
    \frac{\mL(f, K_X, S)(X) - \mL(f, K_X + \varepsilon \phi, S)(X)}{\varepsilon} \right] 
    = h^\top \mathcal{J}_{K_X} \mL(f, K_X, S) \phi \\[1em]
    &\underset{(4)}{\iff} \lim_{\varepsilon \to 0} \frac{1}{\varepsilon} 
    \mathbb{E}_{H_X} \left[ \mL(f, K_X, S)(X) - \mL(f, K_X + \varepsilon \phi, S)(X) \right] 
    = h^\top \mathcal{J}_{K_X} \mL(f, K_X, S) \phi \\[1em]
    &\underset{(5)}{\iff} 0 = h^\top \mathcal{J}_{K_X} \mL(f, K_X, S) \phi, 
    \quad \forall \phi \in W.
\end{aligned}
\]

Where (1) follows by dividing by \(\varepsilon\) and taking the limit as it goes to zero; (2) follows from the continuity of the first derivative and the definition of \(O(\cdot)\); (3) follows from taking expectations under the probability density \(H_X\); (4) follows from interchanging limits and expectations; and (5) follows from the hypothesis.

Finally, the result follows from Proposition~\ref{prop:linear_algebra}, which implies that \(\mathcal{J}_{K_X} \mL(f, K_X, S)) = \mathbf{c} \mathbf{1}^T\) for some \(\mathbf{c}  \in \mathbb{R}^d\).

 \((\Leftarrow)\) If \(\mathcal{J}_{K_X} \mL(f, K_X, S) = \mathbf{c} \mathbf{1}^T\) for any \(K_X\), then by the mean value theorem, substituting into Equation~\ref{expansion:mvt}, we obtain  \(\mL(f, K_X, S) - \mL(f, H_X, S) =  \mathbf{c} \mathbf{1}^T(K_X - H_X) = \mathbf{c}(1 - 1) = 0\). Thus, \(\mL(f, K_X, S) = \mL(f, H_X, S)\) for all \(H_X, K_X\), and consequently,  
 \(\E_{K_X}[\mL(f,K_X,S)(X) - \mL(f,H_X,S)(Xs)] = 0\) for all \(H_X, K_X \in \intsim\).
 
\end{proof}

\subsection{The continuous setting}
\label{sec:continuous_main}

In this section, we introduce the continuous setting and motivate its relevance in a more expository manner; a more formal treatment is provided in the following appendix section.

In section~\ref{sec:discrete_main}, Theorem~\ref{thm:discrete_characterization} and Corollary~\ref{cor:constant_M} show that, in the discrete case, a functional decomposition \(\mL\) that never misattributes effects must be constant with respect to the distribution over covariates.  We now analyze the continuous setting, introducing pertinent regularity assumptions to study how \(\mL\) responds to perturbations in the input distribution.

Namely, we assume \(\mL(f,K_X,S)\) is a continuous functional in its first argument \(f\), Lebesgue measurable in its second argument, \(K_X\), and square integrable, in the \(L^2\) sense, for all triplets \((f,K_X,S)\). 
For example, our first condition is satisfied in cases such as in FANOVA, when \(\mL\) is the integral operator with respect to any probability measure absolutely continuous with respect to the Lebesgue measure. 
The third assumption is identical to those in FANOVA and ALE, which both require \(\mL\) to belong to the space of square integrable functions, \(L^2\). 
Lastly, we assume that the densities \(k(x)\) belong to the space of compactly supported functions, which we denote by \(\prob\). Throughout, we use the notation \(K_X \in \prob\) or \(k(x) \in \prob\) interchangeably to refer to the probability measure and its corresponding density---this abuse of notation will be clear from context.
The definition of ALE already assumes compactly support densities. 
In practice, most distributions can be restricted to a compact region (e.g., age, income, and years of education are all bounded).

We parametrize perturbations around a density \(k(x)\) as \(k(x) + \phi(x)\), for admissible\footnote{We require that \(\phi\) be square integrable and that \(k(x) + \phi\) be a valid probability density; see Appendix~\ref{App:perturbations} for details.} functions \(\phi\). 
We denote by \(\mathcal{D}_{K_X}\) the set of admissible perturbation functions of \(K_X\). 
Throughout, we may write \(k(x) + \phi(x)\) or \(K_X + \phi\) interchangeably to denote such perturbations---again this is a mild abuse of notation and will be clear from context. 
Under an additional condition, assuming that \(\mL\) is continuously differentiable as a function of \(\phi\), we ensure that we can approximate \(\mL(\cdot, \phi)\) with a linear approximation around zero:
\[
    \mL(\cdot,\phi) \approx \mL(\cdot, 0) + D_{\phi}\mL(\cdot, 0)[\phi]. 
\]

Where \(\mL(\cdot,\phi)\) is short notation for \(\mL(f, K_X + \phi, S)\) and \( D_{\phi}\mL(\cdot, 0)\) is the Fréchet derivative of \( \mL \) with respect to the function \( \phi \) evaluated at $\phi = 0$. 
The Fréchet derivative is an operator, and \(D_{\phi}\mL(\cdot, 0)[\phi]\) denotes it taking \(\phi\) as input. 
See Definition~\ref{def:Frechet} and Appendix~\ref{App:perturbations} for a more rigorous discussion of the perturbation functions. 
Although we have not yet verified that FANOVA and ALE satisfy continuous differentiability with respect to perturbations, our conditions are mild, so we conjecture that this is the case. We now attempt to characterize the behavior of the functional \(\mL\) under small perturbations of the distribution \(K_X\).

\begin{theorem}
\label{thm:mean_zero_frechet}
Assume the above regularity conditions on \(\mL\) (see Assumptions \ref{Assumptions:A1} and \ref{Assumptions:A2} in the Appendix). Let \( K_X \in \prob \), and let \(\mathcal{D}_{K_X}\) denote the set of admissible perturbation functions of \(K_X\). If for all \( \phi \in \mathcal{D}_{K_X} \), we have
\[
    \mathbb{E}_{K_X + \phi} \left[ \mL(f, K_X,S)(X) - \mL(f, K_X + \phi, S)(X) \right] = 0,
\]
then
\begin{equation}
    \label{eq:frechet_meanzero}
    \mathbb{E}_{K_X}\left[ D_{\phi}\mL(\cdot, 0)[\phi](X)\right] = 0, \quad \text{for all } \phi \in \mathcal{D}_{K_X}.
\end{equation}
\end{theorem}

See Appendix~\ref{App:proofs} for the proof.

\begin{theorem}
\label{thm:main}
Under the assumptions of Theorem~\ref{thm:mean_zero_frechet}, if a functional decomposition \(\mL(f, K_X, S)\) does not depend on its input distribution (i.e., \( \mL(f, K_X,S) - \mL(f, H_X,S) = 0 \) for all \( f, K_X, H_X, S \)), then it does not misattribute effects of \(Y \mid X\).
\end{theorem}
\begin{proof}
The proof is straightforward: by definition, if \( \mL(f, K_X,S) - \mL(f, H_X,S) = 0 \) for all \( f, K_X, H_X, S \), then \( \Delta(\mathcal{L}, f,K_X,H_X,S) = 0 \).
\end{proof}

We verify that when the KOB decomposition's assumptions are met, it satisfies this theorem. 
Assuming that \( Y \mid X \) remains unchanged and examining Equation~\ref{eqn:KOB}, \( \beta_H = \beta_K \) and the difference in means simplifies to the sum of the covariate effects. 
Since \(\Delta\) depends solely on \( Y \mid X \), its value is zero.

We conjecture a reverse direction of Theorem~\ref{thm:main}, suggesting that under reasonable assumptions, a functional decomposition \(\mL\) will not misattribute effects if it does not depend on its input distribution. 
Specifically, one might hope that when allowing \(K_X\) to range over the probability space, Equation~\ref{eq:frechet_meanzero} would imply that \(D_{\phi}\mL(\cdot, 0)[\phi](x)\) is constant---in a similar way to the discrete case characterized in Theorem~\ref{thm:discrete_characterization}. 
This, in turn, would imply that \(\mL(\cdot, \phi)\) is invariant under perturbations of concentration; in other words, it is locally constant everywhere and, therefore, \(\mL\) does not depend on its input distribution in a meaningful way---analogous to Corollary~\ref{cor:constant_M}.

\begin{conjecture} \label{conj:main_conjecture}
    Under the same regularity conditions as in Theorem~\ref{thm:mean_zero_frechet}. If for all \(K_X \in \prob\) and all \(\phi \in \mathcal{D}_{K_X}\), we have
    \[
        \mathbb{E}_{K_X + \phi} \left[ \mL(f, K_X, S)(x) - \mL(f, K_X + \phi, S)(x) \right] = 0.
    \]
    Then,
    \[
        \mL(f, K_X, S)(x) = \mL(f, H_X, S)(x), \quad \text{for all } K_X, H_X \in \prob.
    \]
\end{conjecture}

While we have not yet fully proved Conjecture~\ref{conj:main_conjecture}, we feel it is intuitively sensible: if a decomposition $\mL$ does not misattributes effects of transport for \emph{any} distribution, then it must be constant with respect to its input distribution.

Our Examples \ref{ex:fanova} and \ref{ex:ale}, together with Section~\ref{sec:gfan_misattribution} and Theorem~\ref{thm:main}, underscore that popular decomposition methods, such as FANOVA and ALE, are not suitable for explaining differences between two populations under Definition~\ref{def:decomp}, highlighting the need to develop novel decomposition techniques to tackle this problem.

\section{Mathematical Framework: The continuous setting}

We now develop the mathematical definitions and assumptions introduced in Appendix~\ref{sec:continuous_main} needed to prove Theorem~\ref{thm:mean_zero_frechet} and to work toward Conjecture~\ref{conj:main_conjecture}. 
We also provide a precise definition of an \emph{admissible perturbation} and show that such perturbations exist for any compactly supported density.

\subsection{Additional notation}
Let \(\Xspace \subseteq \mathbb{R}^d\) be a compact set of possible covariate values, equipped with its Borel \(\sigma\)-algebra \(\mathcal{B}(\Xspace)\). Let \(\mathcal{C}_0(\Xspace)\) denote the set of continuous functions on \(\Xspace\). In what follows, we focus on probability measures whose densities are continuous, strictly positive, and supported on \(\Xspace\). We denote by \(\prob\) the space of such probability measures; formally,
\[
\prob
= \Bigl\{
   P : 
   \forall A \in \mathcal{B}(\Xspace),\;
   P(A) = \int_{A} p(x)\,\mathrm{d}x,\;
   p(x) \in \mathcal{C}_0(\Xspace),\; 
   p(x) > 0 \ \forall x \in \Xspace,\; 
   \int_{\Xspace} p(x)\,\mathrm{d}x = 1
\Bigr\}.
\]
As in Appendix~\ref{Appendix:FANOVA}, we can think of these densities as the Radon-Nikodym derivatives of probability measures that are absolutely continuous with respect to an underlying measure. Since we now focus only on the Lebesgue measure—though our work applies to any underlying measure—we use \(dx\) instead of \(d\lambda(x)\) for clarity.

We make the following regularity and basic assumptions on the functional decomposition \( \mL(f,K_X,S) \).

\begin{assumptions}
\label{Assumptions:A1}
The following hold:
\begin{enumerate}
    \item \textbf{Continuity}: For any \((K_X, S)\), the map \( f \to \mL(f,K_X,S) \) is continuous for almost all \( f \in \measurable \).
    \item \textbf{Measurability}: For any \((f,S)\), the map \( K_X \to \mL(f,K_X,S) \) is Lebesgue measurable for all \(K_X \in \prob\).
    \item \textbf{Integrability}: The map \( (f, K_X, S) \mapsto \mL(f, K_X, S) \) belongs to \( L^2(\Xspace, \lambda) \), for all \( (f, K_X, S) \in \measurable \times \prob\times 2^{[d]} \).
\end{enumerate}
\end{assumptions}

Where we have used the usual notation \(L^2(\Xspace, \lambda)\) to denote the space of square-integrable functions over \(X\) with respect to a measure \(\lambda\), we now omit \(\lambda\) from the notation and write \(L^2(\Xspace)\) to refer specifically to integration with respect to the Lebesgue measure, making the measure explicit otherwise.

\subsection{Admissible perturbation functions}
\label{App:perturbations}

To define the admissible perturbation functions mentioned in Appendix~\ref{sec:continuous_main}, we first need to define Fréchet differentiability.

\begin{definition}[Fr\'echet differentiability; \citet{cheney2001analysis}]
\label{def:Frechet}
Let \( f : D \to Y \) be a mapping from an open set \( D \) in a normed linear space \( X \) into a normed linear space \( Y \). Let \( x \in D \). If there exists a bounded linear map \( A : X \to Y \) such that
\[
\lim_{h \to 0} \frac{\|f(x + h) - f(x) - Ah\|}{\|h\|} = 0,
\]
then \( f \) is said to be \textit{Fréchet differentiable} at \( x \), or simply differentiable at \( x \). Furthermore, \( A \) is called the \textit{Fréchet derivative} of \( f \) at \( x \).
\end{definition} 


\begin{definition}[Admissible perturbation function]
    We say a continuous function \(\phi \in L^2(\Xspace)\) is an admissible perturbation of the probability measure \(K_X\), if \(k(x) + \phi(x)\) is a density of a distribution in \(\prob\) and has full support everywhere \(X\).
\end{definition}
We denote by \(\mathcal{D}_{K_X}\) the set of admissible perturbation functions of \(K_X\): \( \mathcal{D}_{K_X} =\{ \phi \in L^2(\Xspace) : k(x) + \phi(x) > 0 \text{ and } \int_{\Xspace} (k(x) + \phi(x)) \, dx = 1\} \). We show that \(\mathcal{D}_{K_X} \neq \{0\}\) for all \(K_X \in \prob\).

\begin{lemma}
\label{lemma:existence}
    For any distribution \(K_X \in \prob\), there exist an admissible perturbation function different than zero. 
\end{lemma}
\begin{proof}
Let any smooth compactly supported function $\psi \in L^2(\Xspace)$. Then, we can take the define the function
\[
\tilde{\phi}(x) = \psi(x) - \frac{1}{\lebesgueX}\int_{\Xspace}\psi(y)dy
\]
such that \(\tilde{\phi}(x) \in \meanZero\), that is, \(\int_{\Xspace} \tilde{\phi}(x)\, dx\) = 0. To ensure the positivity requirement, we can take a function \(\phi(x) = \varepsilon\tilde{\phi}(x)\), for \(\varepsilon > 0\), which still is in \(L^2(\Xspace)\) and integrates to zero. Such \(\varepsilon > 0\) must satisfy that for a given density \(K_X(x)\), 

\begin{equation}
\label{eq:epsilon_req}
K_X(x)+ \phi(x) = K_X(x)+ \varepsilon\tilde{\phi}(x) > 0 \iff \varepsilon\tilde{\phi}(x) > -K_X(x), \; \forall\; x \in \Xspace.     
\end{equation}

Whenever \(\tilde{\phi}(x) > 0\), Equation~\ref{eq:epsilon_req}
is always satisfied. Thus, the only relevant case is when \(\tilde{\phi}(x) < 0\), for which Equation~\ref{eq:epsilon_req} is satisfied if and only if 
\[
\varepsilon < \frac{-K_X(x)}{\tilde{\phi}(x)},\; \forall \; x \in \Xspace \text{ such that } \tilde{\phi}(x) < 0.
\]
Or equivalently,
\[
\varepsilon \leq \frac{\inf_{x \in \Xspace} K_X(x)}{\sup_{x \in \Xspace} |\Tilde{\phi}(x)|},
\] 
where by assumption the right hand side is strictly positive. Thus \(\phi(x)\) is an admissible perturbation function of \(K_X(x)\).
\end{proof} 

Note that for any fixed density \(K_X\), we can parameterize the functional decomposition in terms of \(\phi(x)\) as follows: \(\mL(f, \phi, S) = \mL(f, K_X + \phi, S) : \measurable \times \mathcal{D}_{K_X}(X) \times 2^{[d]} \to L^{2}(\mathbb{R}^S)\). For this parameterization, in addition to Assumption \ref{Assumptions:A1}, we need to assume the continuous differentiability of \(\mL\) as a function of \(\phi\) (see Assumption \ref{Assumptions:A2}) to ensure that \(\mL\) is Fréchet differentiable as a map from the Banach space \(L^2(\Xspace)\) into the Banach space \(L^2(\mathcal{X_S})\) \citep{zeidlerNonlinearFunctionalAnalysis1986, averbukhTheoryDifferentiationLinear1967}.

\begin{assumptions}
\label{Assumptions:A2}
The map \(\phi \to \mL(\cdot, \phi(x))\) is continuously differentiable as a map from \(L^2(\Xspace)\) into \(L^2(\mathcal{X_S})\).
\end{assumptions}

Under this new assumption, we can linearly approximate \( \mL(\cdot, \phi) \) around \( \phi = 0 \) with a linear and bounded functional.

\[
\mL(\cdot,\phi) = \mL(\cdot, 0) + D_{\phi}\mL(\cdot, 0)[\phi] + o(\|\phi\|_{L^2}).
\]

Where \( D_{\phi}\mL(\cdot, 0)[\phi] \) is the Fréchet  derivative of \( \mL \) with respect to the function \( \phi \) evaluated at the zero function and \( o(\|\phi\|_{L^2}) \) represents a higher-order functional that vanishes faster than \( \|\phi\|_{L^2} \) as \( \phi \to 0 \). More formally, for any \(\delta > 0\), there exists a \(\tau > 0\) such that if \(\|\phi\|_{L^2}< \tau\), then \(|o(\|\phi\|_{L^2})| \leq \delta\|\phi\|_{L^2}\).

\begin{remark}
    The Fr\'echet derivative is a linear and bounded functional which operates on functions \(\phi \in L^2(\Xspace)\). That is, there exist a constant \(C > 0\) such that,
    \[
    \|D_{\phi}\mL(\cdot, 0)[\phi]\|_{L^2} \leq C \|\phi\|_{L^2}.
    \]
\end{remark}

\subsection{Proof of Theorem~\ref{thm:mean_zero_frechet}}
\label{App:proofs}
We first show some lemmas that will be useful through the proof of Theorem~\ref{thm:mean_zero_frechet}.

\begin{lemma}
\label{lemma:integrability}
Given our assumptions, for any \( K_X \in \prob \text{ and } \phi \in \meanZero\), the following integrals are finite. 
\begin{equation}
\label{eq:integrability_1}
\left| \int_{\Xspace} \left(D_{\phi}\mL(\cdot, 0)[\phi](x)\right)\,\phi(x)\, dx \right| < \infty,
\end{equation}
\begin{equation}
\label{eq:integrability_2}
\left| \int_{\Xspace} (D_{\phi}\mL(\cdot, 0)[\phi](x)) \, \phi(x) \, k(x)\, dx \right| < \infty.
\end{equation}

Furthermore, 
\begin{equation}
\label{eq:integrability_3}
\bigg| \int_{\Xspace} o(\|\phi\|_{L_2})(x)(k(x) + \phi(x)) \, dx \bigg| = o(\|\phi\|_{L^2})\end{equation}

\end{lemma}
\begin{proof}
\(k(x)\) is continuous and compactly supported on \(X\), then by a direct consequence of the extreme value theorem, it is bounded: there exists a \(B > 0\) such that \(\sup_{x\in \Xspace}|k(x)|\leq B_k < \infty\); by a similar argument, \(\sup_{x\in \Xspace}|\phi(x)|\leq B_{\phi} < \infty\). We first show Equation~\ref{eq:integrability_1}: 
\begin{equation*}
    \begin{aligned}
        \left| \int_{\Xspace} (D_{\phi} \mL(\cdot; 0)[\phi](x)) \, k(x) \, dx \right|
        &\leq \int_{\Xspace} \left| D_{\phi} \mL(\cdot; 0)[\phi](x) \right| k(x) \, dx \\
        &\leq \left( \int_{\Xspace} \left( D_{\phi} \mL(\cdot; 0)[\phi](x) \right)^2 \, dx \right)^{1/2}
        \left( \int_{\Xspace} k(x)^2 \, dx \right)^{1/2} \\
        &\leq C \cdot \|\phi\|_{L^2} \cdot B_K \cdot \sqrt{\lebesgueX} \\
        & \leq C\cdot B_{\phi}\cdot B_K \cdot \lebesgueX < \infty.
    \end{aligned}
\end{equation*}

To show Equation~\ref{eq:integrability_2}: 
\begin{equation*}
    \begin{aligned}
        \left| \int_{\Xspace} (D_{\phi} \mL(\cdot; 0)[\phi](x)) \, \phi(x) \, dx \right|
        &\leq \int_{\Xspace} \left|D_{\phi} \mL(\cdot; 0)[\phi](x) \right| |\phi(x)| \, dx \\
        &\leq \left( \int_{\Xspace} \left( D_{\phi} \mL(\cdot; 0)[\phi](x) \right)^2 \, dx \right)^{1/2}
        \left( \int_{\Xspace} \phi(x)^2 \, dx \right)^{1/2} \\
        & \leq C \cdot \|\phi\|_{L^2}\cdot B_{\phi} \sqrt{\lebesgueX} \\
        & \leq B_{\phi}^2 \cdot C \cdot \lebesgueX.
    \end{aligned}
\end{equation*}

To show Equation~\ref{eq:integrability_3}: For any \(\delta > 0 \), there exist a \(\tau > 0\) such that if \(\|\phi\|_{L_2} < \tau\), then \(o(\|\phi\|_{L^2}) \leq \delta \|\phi\|_{L_2}\), thus:
\begin{equation*}
\begin{aligned}
\bigg| \int_{\Xspace} o(\|\phi\|_{L_2})(x)(k(x) + \phi(x)) \, dx \bigg| & \leq \int_{\Xspace} |o(\|\phi\|_{L_2})(x)|(k(x) + \phi(x)) \, dx \\
& \leq (B_K + B_\phi) \int_{\Xspace} |o(\|\phi\|_{L_2})| \, dx \\
& \leq (B_K + B_\phi) \cdot \delta \|\phi\|_{L_2} \lebesgueX \\
& = o(\|\phi\|_{L_2}).
\end{aligned}
\end{equation*}
\end{proof}





\begin{lemma}
\label{lemma:orthogonalzero}
Let \( \Xspace \subseteq \mathbb{R}^d \) be a measurable set with finite Lebesgue measure \( \lebesgueX < \infty \). Then, the orthogonal complement of \( \meanZero \) in \( L^2(\Xspace) \) is the space of constant functions on \( \Xspace \); that is,
\[
\meanZero^\perp = \left\{ f \in L^2(\Xspace) : f(x) = c, \text{ a.e. on } \Xspace \right\}.
\]
\end{lemma}
\begin{proof}
Let
\[
V = \left\{ f \in L^2(\Xspace) : f(x) = c, \text{ a.e. on } \Xspace \right\}.
\]

We will prove that \( \meanZero^{\perp} = V \) by first showing that \( V \subseteq \meanZero^{\perp} \). Let \( f \in V \); then, for any \( \psi \in \meanZero \), we have
\[
\int_{\Xspace} f(x) \psi(x) \, dx = c \int_{\Xspace} \psi(x) \, dx = 0.
\]

It remains to show that \( \meanZero^{\perp} \subseteq V\). Let \(f \in \meanZero^{\perp} \), then \(\int_{\Xspace} f(x) \psi(x) \, dx = 0 \)
for any \( \psi(x) \in \meanZero \). In particular, we can take an arbitrary measurable set \( A \subseteq \Xspace \) and define
\[
\psi_A(x) = \chi_A(x) - \frac{\lambda(A)}{\lambda(\Xspace)}
\]

where \( \chi_A(x) \) is the indicator function over \( A \) and \( \lambda \) is the Lebesgue measure. Thus,

\[
0 = \int_{\Xspace} f(x) \psi_A(x) \, dx = \int_{\Xspace} f(x) \chi_A(x) \, dx - \int_{\Xspace} f(x) \frac{\lambda(A)}{\lambda({\Xspace)}} \, dx
\]
\begin{equation}
\label{eq:constant_rn}
\Leftrightarrow \int_A f(x) \, dx = \int_{\Xspace} f(x) \frac{\lambda(A)}{\lambda({\Xspace)}} \, dx = \lambda(A) \left(\frac{\int_{\Xspace} f(x) \, dx}{\lambda({\Xspace)}} 
\right)
\end{equation}

Define \(\mu(A) = \int_A f(x) \, dx\), which is a signed measure absolutely continuous with respect to the Lebesgue measure. 
On one hand, by the Radon-Nikodym Theorem for signed measures (\citet{follandRealAnalysis1999}; Theorem \(3.8
\)), \(f(x)\) is the Lebesgue integrable Radon-Nikodym derivative. On the other, by Equation~\ref{eq:constant_rn}:
\begin{equation}
\label{eq:constant}
\mu(A) = \lambda(A) \cdot c, \text{ for any measurable set A \(\subseteq \Xspace\)},
\end{equation}

where \(c = \left(\frac{\int_{\Xspace} f(x) \, dx}{\lebesgueX} \right)\). 
By the Lebesgue almost everywhere uniqueness of the Radon-Nikodym derivative, we have form Equation~\ref{eq:constant} and definition of \(\mu\) that 
\[
f(x) = c, \text{ a.e. } x \in \Xspace.
\]

Therefore, \(f \in V\) and \(\meanZero^{\perp} \subseteq V\).
\end{proof}

We can now proceed to prove Theorem~\ref{thm:mean_zero_frechet}, which we hope to use in proving our main Conjecture~\ref{conj:main_conjecture} in future work.

\begin{proof}[Proof of Theorem~\ref{thm:mean_zero_frechet}]
By assumption \(\mathbb{E}_{K_X + \phi} \left[ \mL(f, K_X,S)(X) - \mL(f, K_X + \phi, S)(X) \right] = 0, \text{ for all } \phi \in \mathcal{D}_{K_X}.\) i.e.,
\begin{align*}
    0 & = \int_{\Xspace} \left( \mL(f,K_X,S)(x) - \mL(f, K_X + \phi, S)(x) \right) (k(x) + \phi(x)) \, dx,\\
    & = \int_{\Xspace} \left(\mL(\cdot,0)(x) - \mL(\cdot, \phi)(x)\right) (k(x) + \phi(x)) \, dx, \\
    & = -\int_{\Xspace}\left[ D_{\phi}\mL(\cdot, 0)[\phi](x) + o(\|\phi\|_{L^2}(x))\right](k(x) + \phi(x)) \, dx,
\end{align*}
\begin{equation}
\label{eq:exp_H}
    \iff 0 = \int_{\Xspace}\left[ D_{\phi}\mL(\cdot, 0)[\phi](x) + o(\|\phi\|_{L^2})(x)\right] (k(x) + \phi(x))\, dx.
\end{equation}

Then, by Lemma~\ref{lemma:integrability}, we can split the integrals, and rewrite Equation~\ref{eq:exp_H} as:

\begin{equation*}
    \int_{\Xspace}\left(D_{\phi}\mL(\cdot, 0)[\phi](x)\right) k(x) \, dx + \int_{\Xspace}(D_{\phi}\mL(\cdot, 0)[\phi](x)) \phi(x) \, dx = - \int_{\Xspace} o(\|\phi\|_{L^2})(x)(k(x) + \phi(x))\, dx.
\end{equation*}

Since this equation must hold for all \( \phi \in \mathcal{D}_{K_X} \), we can proceed as in the proof of Lemma~\ref{lemma:existence}. 
Specifically, let \( \phi(x) = \varepsilon \psi(x) \) for sufficiently small \( \varepsilon > 0 \) and \( \psi(x) \in \meanZero \). 
Furthermore, by Lemma~\ref{lemma:integrability}, we know the following:
\( \int_{\Xspace} o(\|\phi\|_{L^2})(x)(k(x) + \phi(x)) \, dx = o(\|\phi\|_{L^2})\). 
Note also that \( o(\|\varepsilon \psi\|_{L^2}) = o(\varepsilon \|\psi\|_{L^2}) = o(\varepsilon) \) since \( \|\psi\|_{L^2} < \infty \), then the above equation simplifies to:

\begin{equation*}
\int_{\Xspace}\left(D_{\phi}\mL(\cdot, 0)[\varepsilon\psi](x)\right) k(x) \, dx + \int_{\Xspace}(D_{\phi}\mL(\cdot, 0)[\varepsilon\psi](x)) \varepsilon\psi(x) \, dx = o(\varepsilon).
\end{equation*}

Where by \(o(\varepsilon)\), we mean a constant that goes to zero faster than \(\varepsilon\). 
By the linearity of the Fr\'echet derivative, we can take \( \varepsilon \) out of the operator, divide by it, and since \( \frac{o(\varepsilon)}{\varepsilon} = o(1) \), we obtain:

\begin{equation*}
    \int_{\Xspace}\left(D_{\phi}\mL(\cdot, 0)[\psi](x)\right) k(x) \, dx + \varepsilon\int_{\Xspace}(D_{\phi}\mL(\cdot, 0)[\psi](x)) \psi(x) \, dx = o(1).
\end{equation*}

Taking \(\varepsilon \to 0\), we get that the first integral is equal to zero:
\begin{equation}
\label{eq:cond_1}
    \int_{\Xspace}\left(D_{\phi}\mL(\cdot, 0)[\psi](x)\right) k(x) \, dx = \E_{K_X}[D_{\phi}\mL(\cdot, 0)[\psi](x)] = 0, \quad \text{ for all } \psi \in \meanZero.
\end{equation}

\end{proof}

\end{document}